\documentclass[12pt,a4paper,amsmath,amssymb]{amsart}

\usepackage[latin1]{inputenc}
\usepackage{amsthm}
\usepackage{latexsym}
\usepackage{amsfonts}
\usepackage{bbm,bm,dsfont}
\usepackage{eufrak}
\usepackage{color} 
\usepackage{graphicx}

\newtheorem{theorem}{Theorem}
\newtheorem{lemma}{Lemma}
\newtheorem{corollary}{Corollary}
\newtheorem{proposition}{Proposition}

\theoremstyle{definition}
\newtheorem{remark}{Remark}
\newtheorem{example}{Example}
\newtheorem{definition}{Definition}

\newcommand{\mc}[1]{\mathcal{#1}}

\newcommand{\N}{\mathbb N}
\newcommand{\R}{\mathbb R}
\newcommand{\Z}{\mathbb Z}
\newcommand{\C}{\mathbb C}

\newcommand{\T}{\mathbb T}
\newcommand{\F}{\mathbb F}

\newcommand{\hil}{\mathcal{H}} 

\newcommand{\tr}[1]{\mathrm{tr}\left[#1\right]} 
\def\<{\langle}
\def\>{\rangle}

\newcommand{\id}{\mathbbm{1}} 
\newcommand{\fii}{\varphi}

\newcommand{\eps}{\varepsilon}

\newcommand{\ovl}{\overline}

\begin{document}

\title{Compatibility of covariant quantum channels\\ with emphasis on Weyl symmetry}

\author{Erkka Haapasalo}
\address{Department of Physics and Center for Field Theory and Particle Physics, Fudan University, Shanghai 200433, China}
\email{erkka\_haapasalo@fudan.edu.cn}

\maketitle

\begin{abstract}
Compatibility conditions of quantum channels featuring symmetry through covariance are studied. Compatibility here means the possibility of obtaining two or more channels through partial trace out of a broadcasting channel. We see that covariance conditions can be used to simplify compatibility conditions as the broadcasting channel can be assumed to be covariant in a particular way. A particular emphasis is on Weyl covariance and in determining compatibility conditions for Weyl-covariant channels. The concrete examples studied include the case of a non-compact continuous phase space and the case of a finite phase space.

\vspace{10pt}
\end{abstract}

\section{Introduction}

One of the features of quantum theory which is responsible for its radical deviation from classical physics is the fact that quantum measurement processes cannot typically be realized simultaneously with a joint device. This notion of {\it incompatibility} is widely studied for quantum observables (positive-operator-valued measures) \cite{Busch86,BuHeSchuSte2013,HeKiSchu2015,
HeReSta2008,Lahti2003,LaPu97} and its negation in this context is often called as {\it joint measurability}. Incompatibility can be seen as a non-classical resource since there is an operational connection between the incompatibility of observables and detecting EPR-steering \cite{QuiVeBru2014,GuMoUo2014}. Thus identifying incompatibility is important. In addition to observables, incompatibility has been studied for a wider range of quantum devices \cite{Haapasalo2015,HeMi2017,HeReRyZi2018}.

Since incompatibility is defined through negation, the impossibility of joining, we must first be able to characterize the compatible device pairs or larger assemblages. This is, in general, a demanding task. However, in the presence of symmetry properties the problem simplifies. Furthermore, symmetry is a pervasive feature in most physical measuring devices \cite{VaradarajanKirja} and it is prudent to use this for our advantage.

This paper concentrates on covariant quantum channels and their compatibility properties. Quantum channels are an inalienable part of quantum measurements \cite{kirja} since describe changes of the quantum system in measurement processes or through time evolution. Compatibility of channels means that the channels have a broadcasting channel (a joint channel) from which the channels can be obtained as margins. Hence, compatibility of channels is closely connected to the concept of approximate quantum cloning. In this work, we utilize the symmetry properties of channels and show that in most physical situations, symmetry properties simplify determining compatible channels. This is due to the fact that the joint channels of covariant channels can be restricted to the smaller set of covariant joint channels. Similar methods and results have been successfully used before in the study of joint measurability \cite{Werner2004}. As particular examples, we consider channels covariant under phase space shifts, the Weyl-covariant channels.

The paper is organized as follows: In Section \ref{sec:def}, general notations are fixed and definitions for channels and their covariance are given. We go on to define compatibility in Section \ref{sec:compcov} and give a result stating that, in varied situations, joint channels of covariant channels can be restricted to covariant joint channels. Some compatibility conditions for covariant channels are presented in Subsection \ref{subsec:compdilat}. Our main results deal with phase-space-covariant channels and their compatibility conditions which are discussed in the case of a generalized phase space in Section \ref{sec:gPhaseSpace}. The structure of covariant joint channels for Weyl covariance is discussed in Subsection \ref{subsec:Wjoin} and in Subsection \ref{subsec:Wcompdilat} compatibility conditions for Weyl-covariant channels are given. In Section \ref{sec:PhysPhase}, we discuss the consequences of the preceding section for two types of physical phase spaces, for a non-compact continuous phase space in Subsection \ref{sec:contphasespace} and for finite phase space in Subsection \ref{subsec:finitephase}. In the former case, a particular emphasis is on Gaussian Weyl-covariant channels. Finally in Section \ref{sec:multipartite} some of the earlier results are generalized to the multipartite case.

\section{Definitions}\label{sec:def}

We let $\N$ be the set $\{1,\,2,\ldots\}$ of positive natural numbers. The Hilbert spaces in this treatise are always assumed to be complex. Any additional assumptions on the Hilbert spaces, such as separability, will be explicitly stated in each case. For any Hilbert space $\hil$, we denote by $\mc L(\hil)$, $\mc T(\hil)$, and respectively $\mc U(\hil)$ the sets of bounded linear operators, trace-class operators, and respectively unitary operators on $\hil$. We denote by $\id_\hil$ the identity operator on $\hil$; we often omit the subscript when there is no risk of confusion. Of special interest in quantum physics is the set $\mc S(\hil)\subset\mc T(\hil)$ of positive trace-1 operators on $\hil$, since these operators represent the {\it states} of the quantum system described by $\hil$. Moreover, we denote the inner product, which is linear in the second argument throughout this treatise, of any Hilbert space by $\<\cdot|\cdot\>$; there should be no confusion regarding which Hilbert space is in question in each situation.

Let $\hil$ and $\mc K$ be Hilbert spaces. We say that a linear map $\Phi:\mc L(\mc K)\to\mc L(\hil)$ is {\it completely positive} if for any $n\in\N$, $B_1,\ldots,\,B_n\in\mc L(\mc K)$, and $\fii_1,\ldots,\,\fii_n\in\hil$
$$
\sum_{i,j=1}^n\<\fii_i|\Phi(B_i^*B_j)\fii_j\>\geq0.
$$
If $\Phi(\id_{\mc K})=\id_\hil$, we say that $\Phi$ is {\it unital}. If $\Phi$ is (completely) positive and for any increasing net $(B_\alpha)_{\alpha\in A}\subset\mc L(\mc K)$ of self-adjoined operators we have
$$
\sup_{\alpha\in A}\Phi(B_\alpha)=\Phi\big(\sup_{\alpha\in A} B_\alpha\big)
$$
(or, equivalently, $\Phi$ is (completely) positive and continuous with respect to the ultraweak topologies of $\mc L(\mc K)$ and $\mc L(\hil)$), we say that $\Phi$ is {\it normal}.
\begin{definition}
A normal unital completely positive linear map $\Phi:\mc L(\mc K)\to\mc L(\hil)$ is called a {\it channel} and the set of such channels is denoted by ${\bf Ch}(\hil,\mc K)$.
\end{definition}

Using the duality $\mc T(\hil)^*=\mc L(\hil)$, we may define the predual $\Phi_*:\mc T(\hil)\to\mc T(\mc K)$ for any channel $\Phi\in{\bf Ch}(\hil,\mc K)$ through
$$
\tr{\Phi_*(T)B}=\tr{T\Phi(B)},\qquad T\in\mc T(\hil),\quad B\in\mc L(\mc K).
$$
Clearly $\Phi_*$ is fully determined by its restriction on $\mc S(\hil)$ and $\Phi_*(\rho)\in\mc S(\mc K)$ for any $\rho\in\mc S(\hil)$. Indeed, the physical relevance of channels is in that they describe transformations of quantum states. From the physical point of view, complete-positivity arises from the requirement that trivially extended maps should be positive as well. This means that, for any $n\in\N$, the map $\Phi\otimes{\rm id}_{\mc M_{n\times n}(\C)}:\mc L(\mc K)\otimes\mc M_{n\times n}(\C)\to\mc L(\hil)\otimes\mc M_{n\times n}(\C)$ is required to be positive. Here we denote by $\mc M_{n\times n}(\C)$ the algebra of ($n\times n$)-matrices with complex entries and by ${\rm id}_{\mc M_{n\times n}(\C)}$ the identity map on this algebra. Moreover, the tensor product channel is the extension of the map $\mc L(\mc K)\times\mc M_{n\times n}(\C)\ni(B,M)\mapsto\big(\Phi(B),M\big)\in\mc L(\hil)\times\mc M_{n\times n}(\C)$.

Let $G$ be a group and fix a Hilbert space $\hil$. A map $U:G\to\mc U(\hil)$ is called a {\it projective unitary representation} if the map $G\ni g\mapsto\alpha^U_g\in{\rm Aut}\big(\mc L(\hil)\big)$, $\alpha^U_g(A)=U(g)AU(g)^*$, $g\in G$, $A\in\mc L(\hil)$, is a group homomorphism. This can be shown to be equivalent with
\begin{itemize}
\item $U(e)=\id_\hil$ (with $e$ being the neutral element of $G$) and
\item $U(gh)=m(g,h)U(g)U(h)$ for all $g,\,h\in G$ where $m:G\times G\to\T$ (with $\T$ being the torus, the set of modulus-1 complex numbers) is a {\it $G$-multiplier}, i.e.,
\begin{itemize}
\item $m(g,e)=1=m(e,g)$ for all $g\in G$ and
\item $m(g,h)m(gh,k)=m(g,hk)m(h,k)$ for all $g,\,h,\,k\in G$.
\end{itemize}
\end{itemize}
If the multiplier $m$ above has the constant value 1, $U$ is often called a(n ordinary) unitary representation.

Fix now Hilbert spaces $\hil$ and $\mc K$, group $G$ and projective representations $U:G\to\mc U(\hil)$ and $V:G\to\mc U(\mc K)$. The following definition of covariant channels is the same as, e.g.,\ in \cite{HaPe2017}:
\begin{definition}
We denote by ${\bf Ch}_U^V$ the subset of channels $\Phi\in{\bf Ch}(\hil,\mc K)$ such that
$$
\Phi\circ\alpha^V_g=\alpha^U_g\circ\Phi,\qquad g\in G.
$$
We call channels $\Phi\in{\bf Ch}_U^V$ {\it $(U,V)$-covariant}.
\end{definition}

In typical situations the symmetry group $G$ is locally compact and second countable (lcsc) and the representations are strongly continuous, i.e.,\ the maps $g\mapsto U(g)\fii$, $\fii\in\hil$, are continuous with respect to the locally compact topology of $G$ and the natural topology of $\hil$, and similarly for $V$.

Let $G$ be a locally compact group and $L^\infty(G)$ be the space of (equivalence classes of) essentially bounded functions measurable with respect to any of the equivalent left (or right) Haar measure of $G$. Let us concentrate on the left-invariant case. A locally compact group $G$ is said to be {\it amenable} if it allows an {\it invariant mean} $L^\infty(G)\ni f\mapsto\ovl f\in\C$, i.e.,\ $f\mapsto\ovl f$ is a positive linear functional (i.e.,\ linear map such that $\ovl f\geq0$ whenever $f\geq0$ ,i.e.,\ $f$ is a pointwise-positive function) such that $\ovl 1=1$ (where $1$ on LHS stands for the function having the constant value 1), and, for any $f\in L^\infty(G)$ and $g\in G$, $\ovl{f^g}=\ovl f$ where $f^g\in L^\infty(G)$, $f^g(h)=f(g^{-1}h)$ for all $h\in G$. Compact groups are amenable. Indeed, when we fix the left Haar measure $\mu$ of a compact $G$ with $\mu(G)=1$, the map
$$
L^\infty(G)\ni f\mapsto\ovl f=\int_G f\,d\mu\in\C
$$
is an invariant mean. According to the Markov-Kakutani fixed point theorem, also Abelian groups are amenable. In our most important examples, the symmetry groups will be amenable and invariant means will be used to make general channels covariant.

\section{Compatibility of covariant channels}\label{sec:compcov}

Let $\hil$, $\mc K_1$, and $\mc K_2$ be Hilbert spaces. The following is a central definition of this work:
\begin{definition}\label{def:margincomp}
Let $\Psi\in{\bf Ch}(\hil,\mc K_1\otimes\mc K_2)$. Define the {\it margins} $\Psi_{(i)}\in{\bf Ch}(\hil,\mc K_i)$, $i=1,\,2$, of $\Psi$ through
$$
\Psi_{(1)}(A)=\Psi(A\otimes\id_{\mc K_2}),\quad\Psi_{(2)}(B)=\Psi(\id_{\mc K_1}\otimes B),\qquad A\in\mc L(\mc K_1),\quad B\in\mc L(\mc K_2).
$$
Let $\Phi_i\in{\bf Ch}(\hil,\mc K_i)$, $i=1,\,2$. The channels $\Phi_1$ and $\Phi_2$ are said to be {\it compatible} if there is a {\it joint channel} $\Psi\in{\bf Ch}(\hil,\mc K_1\otimes\mc K_2)$ for $\Phi_1$ and $\Phi_2$, i.e.,\ $\Phi_i=\Psi_{(i)}$, $i=1,\,2$.
\end{definition}

Suppose that $\Phi_i\in{\bf Ch}(\hil,\mc K_i)$, $i=1,\,2$ are compatible and have a joint channel $\Psi\in{\bf Ch}(\hil,\mc K_1\otimes\mc K_2)$. Let us denote the partial trace over $\mc K_i$ on $\mc T(\mc K_1\otimes\mc K_2)$ by ${\rm tr}_i$, $i=1,\,2$. For any $\rho\in\mc S(\hil)$, we have
$$
\Phi_{1\,*}(\rho)=\Psi_{(1)\,*}(\rho)={\rm tr}_2\big[\Psi_*(\rho)\big]\quad\Phi_{2\,*}(\rho)=\Psi_{(2)\,*}(\rho)={\rm tr}_1\big[\Psi_*(\rho)\big].
$$
Thus, compatibility of two channels means that the channels can be seen as reduced dynamics of a broadcasting Schr\"odinger channel $\Psi_*:\mc S(\hil)\to\mc S(\mc K_1\otimes\mc K_2)$.

\begin{remark}\label{rem:0resource}
Let us make a couple of observations on operations that preserve compatibility. Pick compatible channels $\Phi_i\in{\bf Ch}(\hil,\mc K_i)$, $i=1,\,2$. Let us assume that $\Psi\in{\bf Ch}(\hil,\mc K_1\otimes\mc K_2)$ is a joint channel for $\Phi_1$ and $\Phi_2$. Channels of the form $\Phi'_i=\Gamma\circ\Phi_i\in{\bf Ch}(\hil',\mc K_i)$, $i=1,\,2$,  where $\hil'$ is some Hilbert space and $\Gamma\in{\bf Ch}(\hil,\hil')$, are {\it pre-processings of $\Phi_1$ and $\Phi_2$ with a common pre-processor}. Such pre-processings are compatible as well. Indeed, the channel $\Gamma\circ\Psi$ is easily seen to be a joint channel for $\Phi'_1$ and $\Phi'_2$.

Whenever $\mc K_i'$, $i=1,\,2$, are some Hilbert spaces and $\Delta_i\in{\bf Ch}(\mc K_i,\mc K_i')$ are channels, the channels $\tilde{\Phi}_i=\Phi_i\circ\Delta_i\in{\bf Ch}(\hil,\mc K_i')$ are called as {\it post-processings of $\Phi_1$ and $\Phi_2$}. Such post-processings are also compatible. To see this, define the channel $\Delta_1\otimes\Delta_2\in{\bf Ch}(\mc K_1\otimes\mc K_2,\mc K_1'\otimes\mc K_2')$ as the unique extension of the map $\mc L(\mc K_1')\times\mc L(\mc K_2')\ni(C,D)\mapsto\Delta_1(C)\otimes\Delta_2(D)\in\mc L(\mc K_1\otimes\mc K_2)$ onto $\mc L(\mc K_1'\otimes\mc K_2')$. It is immediate that $\Psi\circ(\Delta_1\otimes\Delta_2)$ is a joint channel for $\tilde{\Phi}_1$ and $\tilde{\Phi}_2$. Thus, the set of compatible pairs of channels is closed under pre-processing with a common pre-processor and under post-processing.
\end{remark}

The following lemma will be crucial in proving the main result, Proposition \ref{prop:covjoint}, of this section.

\begin{lemma}\label{lemma:jointnormal}
Let $\hil$, $\mc K_1$, and $\mc K_2$ be Hilbert spaces. Suppose that $\Psi:\mc L(\mc K_1\otimes\mc K_2)\to\mc L(\hil)$ is a positive linear map whose margins $\Psi_{(i)}$, $i=1,\,2$, are defined as in Definition \ref{def:margincomp}. If $\Psi_{(i)}$, $i=1,\,2$ are normal, so is $\Psi$.
\end{lemma}

\begin{proof}
Denote by $\F_i$ the set of finite-dimensional linear subspaces of $\mc K_i$, $i=1,\,2$. These are directed sets with respect to set inclusion. Consequently, the set $\F_1\times\F_2$ is directed in a natural way. Denote by $P_{\mc M}$ the orthogonal projection onto the subspace $\mc M\in\F_i$, $i=1,\,2$. Using the normality of $\Phi_1$ and $\Phi_2$,
\begin{eqnarray*}
\Psi(\id_{\mc K_1\otimes\mc K_2}-P_{\mc M}\otimes P_{\mc N})&=&\Psi\big(\id_{\mc K_1\otimes\mc K_2}-P_{\mc M}\otimes\id_{\mc K_2}+P_{\mc M}\otimes\id_{\mc K_2}(\id_{\mc K_1\otimes\mc K_2}-\id_{\mc K_1}\otimes P_{\mc N})\big)\\
&\leq&\Psi\big((\id_{\mc K_1\otimes\mc K_2}-P_{\mc M}\otimes\id_{\mc K_2})+(\id_{\mc K_1\otimes\mc K_2}-\id_{\mc K_1}\otimes P_{\mc N})\big)\\
&=&\Psi_{(1)}(\id_{\mc K_1}-P_{\mc M})+\Psi_{(2)}(\id_{\mc K_2}-P_{\mc N})\underset{\begin{tiny}(\mc M,\mc N)\in\F_1\times\F_2\end{tiny}}{\rightarrow}0.
\end{eqnarray*}

Suppose that $(C_\alpha)_{\alpha\in A}\subset\mc L(\mc K_1\otimes\mc K_2)$ is an increasing net of self-adjoined operators with the supremum $C$. Fix a positive $T\in\mc T(\hil)$, $\eps>0$, and $\alpha_0\in A$. Choose $(\mc M,\mc N)\in\F_1\times\F_2$, so that
$$
\tr{T\Psi(\id_{\mc K_1\otimes\mc K_2}-P_{\mc M}\otimes P_{\mc N})}\leq\min\Big\{\frac{\eps^2}{36\tr{T}\|C-C_{\alpha_0}\|^2},\frac{\eps}{3\|C-C_{\alpha_0}\|}\Big\},
$$
and $\alpha_1\in A$, so that $\tr{T\Psi\big(P_{\mc M}\otimes P_{\mc N}(C-C_\alpha)P_{\mc M}\otimes P_{\mc N}\big)}\leq\eps/3$ for all $\alpha\in A$, $\alpha\geq\alpha_1$; the first choice can be made based on the above estimation and the second one on the fact that $\Psi$ restricted on a finite-dimensional sub-algebra is normal. We may now evaluate for any $\alpha\geq\alpha_0,\,\alpha_1$
\begin{eqnarray*}
&&\tr{T\Psi(C-C_\alpha)}\\
&\leq&\tr{T\Psi(P_{\mc M}\otimes P_{\mc N}(C-C_\alpha)P_{\mc M}\otimes P_{\mc N}}\\
&&+\big|\tr{T\Psi\big((\id_{\mc M_1\otimes\mc M_2}-P_{\mc M}\otimes P_{\mc N})(C-C_\alpha)P_{\mc M}\otimes P_{\mc N}\big)}\big|\\
&&+\big|\tr{T\Psi\big(P_{\mc M}\otimes P_{\mc N}(C-C_\alpha)(\id_{\mc M_1\otimes\mc M_2}-P_{\mc M}\otimes P_{\mc N})\big)}\big|\\
&&+\tr{T\Psi\big((\id_{\mc M_1\otimes\mc M_2}-P_{\mc M}\otimes P_{\mc N}\big)(C-C_\alpha)(\id_{\mc M_1\otimes\mc M_2}-P_{\mc M}\otimes P_{\mc N})\big)}\\
&\leq&\tr{T\Psi(P_{\mc M}\otimes P_{\mc N}(C-C_\alpha)P_{\mc M}\otimes P_{\mc N}}\\
&&+2\sqrt{\tr{T\Psi(\id_{\mc K_1\otimes\mc K_2}-P_{\mc M}\otimes P_{\mc N})}\tr{T\Psi\big(P_{\mc M}\otimes P_{\mc N}(C-C_\alpha)^2P_{\mc M}\otimes P_{\mc N}\big)}}\\
&&+\|C-C_\alpha\|\tr{T\Psi(\id_{\mc K_1\otimes\mc K_2}-P_{\mc M}\otimes P_{\mc N})}\\
&\leq&\tr{T\Psi(P_{\mc M}\otimes P_{\mc N}(C-C_\alpha)P_{\mc M}\otimes P_{\mc N}}+2\|C-C_\alpha\|\sqrt{\tr{T\Psi(\id_{\mc K_1\otimes\mc K_2}-P_{\mc M}\otimes P_{\mc N})}\tr{T}}\\
&&+\|C-C_\alpha\|\tr{T\Psi(\id_{\mc K_1\otimes\mc K_2}-P_{\mc M}\otimes P_{\mc N})}\leq\eps/3+\eps/3+\eps/3=\eps.
\end{eqnarray*}
The second estimation above is based on the Cauchy-Schwarz inequality for the positive sesquilinear form $(C_1,C_2)\mapsto\tr{T\Psi(C_1^*C_2)}$ and on the inequality $D\leq\|D\|\id_{\mc K_1\otimes\mc K_2}$ for any self-adjoined $D\in\mc L(\mc K_1\otimes\mc K_2)$. Note also that, as $(C-C_\alpha)_{\alpha\in A}$ is a descending sequence of positive operators, also $(\|C-C_\alpha\|)_{\alpha\in A}$ is non-increasing, as one easily checks. Thus, $\Psi$ is normal.
\end{proof}

For the rest of this section, we shall fix Hilbert spaces $\hil$, $\mc K_1$, and $\mc K_2$, a symmetry group $G$, and projective representations $U:G\to\mc U(\hil)$ and $V_i:G\to\mc U(\mc K_i)$, $i=1,\,2$. We are interested in the compatibility properties of pairs of channels within ${\bf Ch}_U^{V_1}\times{\bf Ch}_U^{V_2}$. Define the representation $V_{12}:G\to\mc U(\mc K_1\otimes\mc K_2)$, $V_{12}(g)=V_1(g)\otimes V_2(g)$, $g\in G$. It follows easily that, whenever $\Psi\in{\bf Ch}_U^{V_{12}}$, then $\Psi_{(i)}\in{\bf Ch}_U^{V_i}$. However, in general, it does not hold that if $\Phi_i\in{\bf Ch}_U^{V_i}$, $i=1,\,2$, are compatible, they should have a joint channel $\Psi\in{\bf Ch}_U^{V_{12}}$. Next we will show however, that in varied situations a covariant joint channel exists.

\begin{proposition}\label{prop:covjoint}
Suppose that $\hil$ and $\mc K_i$, $i=1,\,2$, are separable Hilbert spaces, $G$ is an amenable locally compact group, and $U:G\to\mc U(\hil)$ and $V_i:G\to\mc U(\mc K_i)$, $i=1,\,2$, are strongly continuous projective unitary representations. Any two compatible covariant channels $\Phi_i\in{\bf Ch}_U^{V_i}$, $i=1,\,2$, have a joint channel $\ovl\Psi\in{\bf Ch}_U^{V_{12}}$.
\end{proposition}

\begin{proof}
Let $\Phi_i\in{\bf Ch}_U^{V_i}$, $i=1,\,2$, be compatible and $\Psi\in{\bf Ch}(\hil,\mc K_1\otimes\mc K_2)$ be a joint channel for them which is not necessarily $(U,V_{12})$-covariant. Next, we shall `covariantize' $\Psi$.

For any $\fii,\,\psi\in\hil$ and $C\in\mc L(\mc K_1\otimes\mc K_2)$, define the function $f_{C,\fii,\psi}:G\to\C$,
$$
f_{C,\fii,\psi}(g)=\<U(g)^*\fii|\Psi\big(V_{12}(g)^*CV_{12}(g)\big)U(g)^*\psi\>,\qquad g\in G.
$$
It follows easily from the separability of the Hilbert spaces and the strong continuity of the representations that $f_{C,\fii,\psi}$ are Haar measurable for all $C\in\mc L(\mc K_1\otimes\mc K_2)$ and $\fii,\,\psi\in\hil$. Moreover,
$$
|f_{C,\fii,\psi}(g)|\leq\|\fii\|\|\psi\|\|C\|,\qquad C\in\mc L(\mc K_1\otimes\mc K_2),\quad\fii,\,\psi\in\hil,\quad g\in G.
$$
Thus $f_{C,\fii,\psi}\in L^\infty(G)$ for all $C\in\mc L(\mc K_1\otimes\mc K_2)$ and $\fii,\,\psi\in\hil$.

Pick an invariant mean $L^\infty(G)\ni f\mapsto\ovl f\in\C$. It can be easily checked that $\ovl{f_{C,\cdot,\cdot}}:\hil\times\hil\to\C$ is a bounded sesquilinear form for every $C\in\mc L(\mc K_1\otimes\mc K_2)$. Thus, there is a map $\ovl\Psi:\mc L(\mc K_1\otimes\mc K_2)\to\mc L(\hil)$ such that
$$
\ovl{f_{C,\fii,\psi}}=\<\fii|\ovl\Psi(C)\psi\>,\qquad C\in\mc L(\mc K_1\otimes\mc K_2),\quad\fii,\,\psi\in\hil.
$$
Reader may easily verify that $\ovl\Psi$ is linear and unital. A simple calculation shows that $f_{C,\fii,\psi}^g=f_{\alpha^{V_{12}}(C),U(g)\fii,U(g)\psi}$ for all $g\in G$, $C\in\mc L(\mc K_1\otimes\mc K_2)$, and $\fii,\,\psi\in\hil$. Using this and the defining feature of the invariant mean, we have
\begin{eqnarray*}
\<\fii|\ovl\Psi\big(V_{12}(g)CV_{12}(g)^*\big)\psi\>=\ovl{f_{\alpha^{V_{12}}(C),\fii,\psi}}=\ovl{f_{C,U(g)^*\fii,U(g)^*\psi}^g}&=&\ovl{f_{C,U(g)^*\fii,U(g)^*\psi}}\\
&=&\<\fii|U(g)\ovl\Psi(C)U(g)^*\psi\>
\end{eqnarray*}
for all $g\in G$, $\fii,\,\psi\in\hil$, and $C\in\mc L(\mc K_1\otimes\mc K_2)$, implying that
$$
\ovl\Psi\circ\alpha^{V_{12}}_g=\alpha^U_g\circ\ovl\Psi,\qquad g\in G.
$$

Fix $n\in\N$ and $\fii_1,\ldots,\,\fii_n\in\hil$ and $C_1,\ldots,\,C_n\in\mc L(\mc K_1\otimes\mc K_2)$. Using the complete positivity of $\Psi$, we have for any $g\in G$
$$
\sum_{i,j=1}^n f_{C_i^*C_j,\fii_i,\fii_j}(g)=\sum_{i,j=1}^n\<U(g)^*\fii_i|\Psi\big(V_{12}(g)^*C_i^*C_jV_{12}(g)\big)U(g)^*\fii_j\>\geq0.
$$
Using the positivity of the invariant mean, it follows that
$$
\sum_{i,j=1}^n\<\fii_i|\ovl\Psi(C_i^*C_j)\fii_j\>=\ovl{\sum_{i,j=1}^n f_{C_i^*C_j,\fii_i,\fii_j}}\geq0.
$$
Define the functions $f_{(1),A,\fii,\psi}\in L^\infty(G)$, $A\in\mc L(\mc K_1)$, $\fii,\,\psi\in\hil$,
$$
f_{(1),A,\fii,\psi}(g)=\<U(g)^*\fii|\Phi_1\big(V_1(g)^*AV_1(g)\big)U(g)^*\psi\>,\qquad g\in G.
$$
It follows from the covariance of $\Phi_1$ that $f_{(1),A,\fii,\psi}$ has the constant value $\<\fii|\Phi_1(A)\psi\>$. Thus, $\ovl{f_{A\otimes\id_{\mc K_2},\fii,\psi}}=\ovl{f_{(1),A,\fii,\psi}}=\<\fii|\Phi_1(A)\psi\>$ for all $A\in\mc L(\mc K_1)$ and $\fii,\,\psi\in\hil$. Hence, for all $\fii,\,\psi\in\hil$ and $A\in\mc L(\mc K_1)$,
$$
\<\fii|\ovl\Psi(A\otimes\id_{\mc K_2})\psi\>=\ovl{f_{A\otimes\id_{\mc K_2},\fii,\psi}}=\<\fii|\Phi_1(A)\psi\>.
$$
Similarly, $\ovl\Psi(\id_{\mc K_1}\otimes B)=\Phi_2(B)$ for all $B\in\mc L(\mc K_2)$. Thus, $\ovl\Psi$ is a covariant unital completely positive linear map whose margins (defined analogously to normal maps) coincide with $\Phi_1$ and $\Phi_2$. The normality of $\ovl\Psi$ now follows immediately from Lemma \ref{lemma:jointnormal}.
\end{proof}

Thus, under the conditions of Proposition \ref{prop:covjoint}, covariant channels are compatible if and only if they have a covariant joint channel. This greatly simplifies the compatibility conditions for covariant channels, as we will see in our examples.

\subsection{Dilation-theoretic conditions for compatibility}\label{subsec:compdilat}

Suppose that $G$ is a group and $U:G\to\mc U(\hil)$ and $V:G\to\mc U(\mc K)$ are projective unitary representations.

\begin{definition}
A triple $(\mc L,J,\ovl U)$ of a Hilbert space $\mc L$, an isometry $J:\hil\to\mc K\otimes\mc L$, and a projective unitary representation $\ovl U:G\to\mc U(\mc L)$ is a {\it covariant Stinespring dilation for a channel $\Phi\in{\bf Ch}_U^V$} if $\Phi(A)=J^*(A\otimes\id_{\mc L})J$, $A\in\mc L(\mc K)$, and $JU(g)=\big(V(g)\otimes\ovl U(g)\big)J$ for all $g\in G$. If additionally the vectors $(A\otimes\id_{\mc L})J\fii$, $A\in\mc L(\mc K)$, $\fii\in\hil$, span a dense subspace in $\mc K$, the covariant dilation $(\mc L,V,\ovl U)$ is called {\it minimal}.
\end{definition}

Any channel $\Phi\in{\bf Ch}_U^V$ possesses a covariant Stinespring dilation and among them there is a minimal one which is unique up to unitary equivalence \cite[Corollary 1]{HaPe2017}. Clearly, if $U$ is associated with the multiplier $m_U$ and $V$ with $m_V$, the projective representation $\ovl U$ in any of the covariant dilations is associated with the multiplier $(g,h)\mapsto m_U(g,h)\ovl{m_V(g,h)}$.

\begin{proposition}\label{prop:covdilat}
Retain the assumptions of Proposition \ref{prop:covjoint}. Let $(\mc L,J,\ovl U)$ be the (essentially unique) covariant minimal Stinespring dilation for $\Phi_1$. The channels $\Phi_1$ and $\Phi_2$ are compatible if and only if there is $\tilde{\Phi}_2\in{\bf Ch}_{\ovl U}^{V_2}$ such that
\begin{equation}\label{eq:covdilat}
\Phi_2(B)=J^*\big(\id_{\mc K_1}\otimes\tilde{\Phi}_2(B)\big)J,\qquad B\in\mc L(\mc K_2).
\end{equation}
\end{proposition}

\begin{proof}
If $\Phi_1$ and $\Phi_2$ are compatible, according to Proposition \ref{prop:covjoint}, they are margins of a covariant joint channel $\Psi\in{\bf Ch}_U^{V_{1,2}}$. The existence of the channel $\tilde{\Phi}_2$ of the claim now follow from \cite[Proposition 1, Remark 2]{HaPe2017}.

Suppose that the channel $\tilde{\Phi}_2\in{\bf Ch}_U^{\ovl U}$ of the claim exists. The map
$$
\mc L(\mc K_1)\times\mc L(\mc K_2)\ni(A,B)\mapsto A\otimes\tilde{\Phi}_2(B)\in\mc L(\mc K_1\otimes\mc L)
$$
uniquely extends into a channel defined on $\mc L(\mc K_1\otimes\mc K_2)$. Thus, we may define the channel $\Psi\in{\bf Ch}(\hil,\mc K_1\otimes\mc K_2)$,
$$
\Psi(A\otimes B)=J^*\big(A\otimes\tilde{\Phi}_2(B)\big)J,\qquad A\in\mc L(\mc K_1),\quad B\in\mc L(\mc K_2).
$$
It is immediate that $\Psi_{(1)}=\Phi_1$, $\Psi_{(2)}=\Phi_2$.
\end{proof}

Note that the above result is symmetric, i.e.,\ we may switch the roles of $\Phi_1$ and $\Phi_2$. One may also verify that the channel $\Psi$ in the latter part of the proof is actually $(U,V_{1,2})$-covariant, although it is not necessary for proving the claim.

\section{General phase space and Weyl covariance}\label{sec:gPhaseSpace}

In this treatise, quantum phase space is modelled on an Abelian locally compact group ${\bf X}$ describing the configuration space of the system and the dual group $\hat{\bf X}$ of characters (continuous group homomorphisms from ${\bf X}$ onto the torus $\T$, the set of modulus-1 complex numbers) describing the momentum space. The momentum space is also a locally compact Abelian group with respect to pointwise multiplication. In notation, we treat ${\bf X}$ as a multiplicative group (with neutral element $1$) and $\hat{\bf X}$ as an additive group (with neutral element $0$). The phase space itself is the group $G:={\bf X}\times\hat{\bf X}$ with the group law $(x,\xi)(y,\upsilon)=(xy,\xi+\upsilon)$. The neutral element of $G$ is thus $e=(1,0)$.

We must place a particular regularity condition on our phase space: we require that there be a continuous map $(\cdot|\cdot):{\bf X}\times\hat{\bf X}\to\R$ such that, for any $x\in{\bf X}$ and $\xi\in\hat{\bf X}$, $(x|\cdot):\hat{\bf X}\to\R$ and $(\cdot|\xi):{\bf X}\to\R$ are group homomorphisms and the dual pairing is given by
$$
\<x,\xi\>=e^{i(x|\xi)},\qquad x\in{\bf X},\quad\xi\in\hat{\bf X}.
$$
This means that a continuous version of the logarithm of the dual pairing can be given in the form of a `scalar product'. This allows us to define the symplectic form $S:G\times G\to\R$,
$$
S(g,h)=(x|\upsilon)-(y|\xi),\qquad g=(x,\xi)\in G,\quad h=(y,\upsilon)\in G.
$$

The (spin-0) particle associated with the phase space $G$ is described by the Hilbert space $\hil:=L^2({\bf X})$, the space of (equivalence classes of) functions $\fii:{\bf X}\to\C$ which are square-integrable with respect to a fixed (and hence any) Haar measure of ${\bf X}$. Translations of the configuration space are mirrored by the representation $U:{\bf X}\to\mc U(\hil)$, $\big(U(x)\fii\big)(y)=\fii(x^{-1}y)$, $x,\,y\in{\bf X}$, $\fii\in\hil$, and translations of the momentum space are represented by $V:\hat{\bf X}\to\mc U(\hil)$, $\big(V(\xi)\fii\big)(x)=\<x,\xi\>\fii(x)$, $x\in{\bf X}$, $\xi\in\hat{\bf X}$, $\fii\in\hil$. The phase space translations manifest themselves through the map $W:G\to\mc U(\hil)$,
$$
W(x,\xi)=e^{\frac{i}{2}(x|\xi)}U(x)V(\xi),\qquad(x,\xi)\in G.
$$
It follows that
$$
W(gh)=e^{\frac{i}{2}S(g,h)}W(g)W(h),\qquad g,\,h\in G,
$$
i.e.,\ $W$ is a projective unitary representation, the so-called {\it Weyl representation}, with the multiplier $(g,h)\mapsto e^{\frac{i}{2}S(g,h)}$. Moreover, we have the canonical commutation relations
\begin{equation}\label{eq:gWproperties}
W(g)^*=W(g^{-1}),\quad W(g)W(h)=e^{-iS(g,h)}W(h)W(g),\qquad g,\,h\in G.
\end{equation}

\begin{proposition}\label{prop:Wirr}
The Weyl representation $W$ is irreducible and operators $W(g)$, $g\in G$, span an ultraweakly dense operator system in $\mc L(\hil)$. Similarly, the projective representation $G\times G\ni (g,h)\mapsto W(g)\otimes W(h)\in\mc U(\hil\otimes\hil)$ is irreducible and the operators $W(g)\otimes W(h)$, $g,\,h\in G$ span an ultraweakly dense operator system in $\mc L(\hil\otimes\hil)$.
\end{proposition}

\begin{proof}
The claims for the latter representation follow immediately from the claims concerning $W$. Thus, we concentrate on the first half of the claim.

To see that $W$ is irreducible, let $B\in\mc L(\hil)$ be such that $BW(g)=W(g)B$ for all $g\in G$. Thus, especially, $BV(\xi)=V(\xi)B$ for all $\xi\in\hat{\bf X}$, implying \cite[Part II, Ch. 2, Sec. 5, Theorem 1]{Dix} that there is an essentially bounded measurable function $b:{\bf X}\to\C$ such that $(B\fii)(x)=b(x)\fii(x)$ for all $\fii\in\hil$ and $x\in{\bf X}$. Similarly, $BU(x)=U(x)B$ for all $x\in{\bf X}$, implying $b(x^{-1}y)=b(y)$ for a.a. $(x,y)\in{\bf X}^2$. Using the Fubini theorem, we easily see that $b$ is constant so that $B$ is a scalar multiple of the identity. Thus, according to the Shur orthogonality condition, $W$ is irreducible.

Denote by $\mc D'$ the commutant of $\mc D\subseteq\mc L(\hil)$, i.e.,\ the set of all operators $E\in\mc L(\hil)$ such that $DE=ED$ for all $D\in\mc D$ and by $\mc D'':=(\mc D')'$ the double commutant of $\mc D\subseteq\mc L(\hil)$. According to the von Neumann density theorem \cite[Part I, Ch. 3, Sec. 4, Theorem 2]{Dix}, the ultraweak closure of the linear hull of the range ${\rm ran}\,W$ of $W$ coincides with the double commutant $({\rm ran}\,W)''=(\C\id_\hil)'=\mc L(\hil)$.
\end{proof}

\begin{remark}\label{rem:neq0}
In the sequel, particularly in the proof of Lemma \ref{lemma:CCRChan}, it is important to guarantee that the supports of the characteristic functions (or Fourier transforms) of states (which are functions on the phase space) cover the whole phase space. Let us prove this. Define the function $\hat{S}:G\to\C$, $\hat{S}(g)=\tr{SW(g)}$, $g\in G$, for every $S\in\mc T(\hil)$. For any compact sets $K_1\subseteq{\bf X}$ and $K_2\subseteq\hat{\bf X}$ there is $S\in\mc T(\hil)$ such that $\hat{S}(g)\neq0$ for all $g\in K_1\times K_2$. To see this, pick compact sets $K_1\subseteq{\bf X}$ and $K_2\subseteq\hat{\bf X}$. Using the continuity of the form $(\cdot|\cdot)$ and the compactness of $K_1$ and $K_2$, it can be shown that, for any fixed (small) $\eps>0$, there are compact environments $\tilde{K}_1\subseteq{\bf X}$ and $\tilde{K}_2\in\hat{\bf X}$ of, respectively, $1$ and $0$ such that $(x|\upsilon-\xi)\in[-\eps,\eps]$ and $(y|\xi)\in[-\eps,\eps]$ for all $(x,\xi)\in\tilde{K}_1\times\tilde{K}_2$ and $(y,\upsilon)\in K_1\times K_2$. Fix a Haar measure $dx$ of ${\bf X}$ and the associated Plancherel measure $d\xi$ for $\hat{\bf X}$. Define $S=|\check{\chi}_{\tilde{K}_2}\>\<\chi_{\tilde{K}_1}|$, where $\chi_K$ is the characteristic function of a set $K$ and $\check{\chi}_{\tilde{K}_2}$ is the inverse Fourier transform
$$
\check{\chi}_{\tilde{K}_2}(x)=\int_{K_2}e^{-i(x|\xi)}\,d\xi,\qquad x\in{\bf X}.
$$
It follows that
$$
\hat{S}(y,\upsilon)=e^{\frac{i}{2}(y|\upsilon)}\int_{\tilde{K}_1}\int_{\tilde{K}_2}e^{i\big((x|\upsilon-\xi)+(y|\xi)\big)}\,d\xi\,dx\approx|\tilde{K}_1|\,|\tilde{K}_2|e^{\frac{i}{2}(y|\upsilon)}\neq0
$$
for all $(y,\upsilon)\in K_1\times K_2$, where $|K|$ is the size of the set $K$ with respect to the appropriate measure.
\end{remark}


\subsection{Weyl-covariant channels and their joinings}\label{subsec:Wjoin}

We are interested in channels $\Phi\in{\bf Ch}_W^W$, i.e.,\ channels that behave covariantly under phase-space translations and compatibility conditions of pairs of such channels. Let us define an additional representation $W^{\otimes2}:G\to\mc U(\hil\otimes\hil)$, $W^{\otimes2}(g)=W(g)\otimes W(g)$, $g\in G$. Since the Hilbert space $\hil$ is separable, the representation $W$ is strongly continuous, and, as a locally compact Abelian group, $G$ is amenable, Proposition \ref{prop:covjoint} implies that two channels $\Phi_i\in{\bf Ch}_W^W$, $i=1,\,2$, are compatible if and only if they are margins of a channel $\Psi\in{\bf Ch}_W^{W^{\otimes2}}$. Such covariant joint channels are good candidates for optimal cloning channels due to the irreducibility of $W$. Proposition \ref{prop:Wirr} implies that $\Phi\in{\bf Ch}(\hil,\hil)$ is fully characterized by the images $\Phi\big(W(g)\big)$, $g\in G$, and $\Psi\in{\bf Ch}(\hil,\hil\otimes\hil)$ is exhaustively determined by the images $\Psi\big(W(g)\otimes W(h)\big)$, $g,\,h\in G$. The following lemma slightly modifying well known results on channels on CCR-algebras \cite{DeVaVe77} will be useful in the sequel.

\begin{lemma}\label{lemma:CCRChan}
Let $G_i={\bf X}_i\times\hat{\bf X}_i$, $i=1,\,2$, be Abelian phase spaces (like the ones introduced above), $S_i:G_i\times G_i\to\R$ be the symplectic forms on $G_i$, $i=1,\,2$, $\hil_i=L^2({\bf X}_i)$, $i=1,\,2$, and $W_i:G_i\to\mc U(\hil_i)$ be the Weyl representations associated with $G_i$, $i=1,\,2$. Suppose that $T:G_2\to G_1$ is a continuous group homomorphism. A map $\Psi:\mc L(\hil_2)\to\mc L(\hil_1)$ defined through
\begin{equation}\label{eq:CCRChan}
\Psi\big(W_2(g)\big)=f(g)[W_1\circ T](g),\qquad g\in G_2,
\end{equation}
with a function $f:G_2\to\C$ is a channel if and only if $f$ is continuous, $f(e_2)=1$ ($e_2$ being the neutral element of $G_2$), and, for every $n\in N$ and all $g_1,\ldots,\,g_n\in G_2$,
\begin{equation}\label{eq:ineqCCRChan}
\bigg(e^{\frac{i}{2}\big[S_2(g_i,g_j)-S_1\big(T(g_i),T(g_j)\big)\big]}f(g_i^{-1}g_j)\bigg)_{i,j=1}^n\geq0.
\end{equation}
\end{lemma}

\begin{proof}
Since $W(g)$, $g\in G_2$, are linearly independent and they span an ultraweakly dense operator system in $\mc L(\hil_2)$, Equation \eqref{eq:CCRChan} defines a linear map. It remains to pinpoint the necessary and sufficient conditions for complete positivity, unitality, and normality the function $f$ has to satisfy.

Pick $n,\,m\in\N$, $\gamma_{i,r}\in\C$, $g_{i,r}\in G_2$, $r=1,\ldots,\,m$, $C_i=\sum_{r=1}^m\gamma_{i,r}W_2(g_{i,r})$, $i=1,\ldots,\,n$. Due to normality, complete positivity of $\Psi$ is fully determined by the conditions
\begin{eqnarray*}
0&\leq&\sum_{i,j=1}^n\<\fii_i|\Psi(C_i^*C_j)\fii_j\>=\sum_{i,j=1}^n\sum_{r,s=1}^m\ovl{\gamma_{i,r}}\gamma_{j,s}\<\fii_i|\Psi\big(W_2(g_{i,r})^*W_2(g_{j,s})\big)\fii_j\>\\
&=&\sum_{i,j=1}^n\sum_{r,s=1}^m\ovl{\gamma_{i,r}}\gamma_{j,s}e^{\frac{i}{2}S_2(g_{i,r},g_{j,s})}f(g_{i,r}^{-1}g_{j,s})\<\fii_i|[W_1\circ T](g_{i,r}^{-1}g_{j,s})\fii_j\>\\
&=&\sum_{i,j=1}^n\sum_{r,s=1}^m\ovl{\gamma_{i,r}}\gamma_{j,s}e^{\frac{i}{2}\big[S_2(g_{i,r},g_{j,s})-S_1\big(T(g_{i,r}),T(g_{j,s})\big)\big]}f(g_{i,r}^{-1}g_{j,s})\<[W_1\circ T](g_{i,r})\fii_i|[W_1\circ T](g_{j,s})\fii_j\>.
\end{eqnarray*}
Careful reading reveals that these conditions equal the inequality \eqref{eq:ineqCCRChan}.

For normality, let us define the Fourier transform \cite{Werner84} $\hat{T}:G_2\to\C$ of a trace-class operator $T\in\mc T(\hil_2)$, $\hat{T}(g)=\tr{TW_2(g)}$, $g\in G$. A counterpart of the Bochner theorem holds \cite[Proposition 3.4]{Werner84}: a function $F:G_2\to\C$ is a Fourier transform of a positive trace-class operator $T\in\mc T(\hil_2)$ if and only if $F$ is continuous and, for every $n\in\N$ and all $g_1,\ldots,\,g_n\in G_2$,
\begin{equation}\label{eq:BochnerPositive}
\bigg(e^{\frac{i}{2}S_2(g_i,g_j)}F(g_i^{-1}g_j)\bigg)_{i,j=1}^n\geq0.
\end{equation}

Normality of $\Psi$ is equivalent with the function $F_S:\,G_2\ni g\mapsto\tr{S\Psi\big(W_2(g)\big)}\in\C$ being the Fourier transform of some positive $T\in\mc T(\hil_2)$ for every positive $S\in\mc T(\hil_1)$. The latter condition of the above Bochner theorem is easily seen to reduce to the property of the inequality \eqref{eq:ineqCCRChan}. It thus remains to require that $F_S$ be continuous for every (positive) $S\in\mc T(\hil_1)$. The result of Remark \ref{rem:neq0} can be used to find that, for any compact $K_1\subseteq{\bf X}_1$ and $K_2\subseteq\hat{\bf X}_1$ there is $S\in\mc T(\hil_1)$ such that $\hat{S}$ is continuous and and  non-zero on $K_1\times K_2$. Since
$$
F_S(g)=f(g)[\hat{S}\circ T](g),\qquad g\in G_2,
$$
and $\hat{S}\circ T$ is continuous for every $S\in\mc T(\hil_1)$, the above (together with the fact that $G$ is locally compact) tells us that $f$ has to be continuous. If $f$ is continuous, the first condition of the modified Bochner theorem is automatically satisfied for any $F_S$, $S\in\mc T(\hil_1)$. 
\end{proof}

Let us continue considering a single phase space $G$ and the associated structures $S$, $\hil$, and $W$. Recall that a function $f:G\to\C$ is of positive type if, for any $n\in\N$ and all $g_1,\ldots,\,g_n\in G$,
$$
\big(f(g_i^{-1}g_j)\big)_{i,j=1}^n\geq0.
$$
Denote by $C_+^0(G)$ the set of continuous functions $f:G\to\C$ of positive type with $f(e)=1$. Channels $\Phi\in{\bf Ch}_W^W$ have been completely characterized \cite{DaFuHo2006,Holevo2005}: For any $\Phi\in{\bf Ch}_W^W$ there is a unique $f_\Phi\in C_+^0(G)$ such that $\Phi\big(W(g)\big)=f_\Phi(g)W(g)$ for all $g\in G$ and the map ${\bf Ch}_W^W\ni\Phi\mapsto f_\Phi\in C_+^0(G)$ is bijective. Let $\Phi\in{\bf Ch}_W^W$ and $\mu_\Phi$ be the probability measure on the Borel $\sigma$-algebra of $G$ whose Fourier transform coincides with $f_\Phi$, i.e.,\ $f_\Phi(g)=\int_G e^{-iS(g,h)}\,d\mu_\Phi(h)$. It follows that, for any $g\in G$,
\begin{eqnarray*}
\Phi\big(W(g)\big)=f_\Phi(g)W(g)=\int_G e^{-iS(g,h)}\,d\mu_\Phi(h)W(g)&=&\int_G e^{-iS(g,h)}W(h)^*W(h)W(g)\,d\mu_\Phi(h)\\
&=&\int_G W(h)^*W(g)W(h)\,d\mu_\Phi(h).
\end{eqnarray*}
Thus, $\Phi$ is also characterized by
$$
\Phi(B)=\int_G W(g)^*BW(g)\,d\mu_\Phi(g),\qquad B\in\mc L(\hil).
$$

As already noted, the compatibility properties of pairs within ${\bf Ch}_W^W\times{\bf Ch}_W^W$ crucially depend upon the set ${\bf Ch}_W^{W^{\otimes2}}$ the characterization task of which we next undertake. We will see that, as in the case of ${\bf Ch}_W^W$, channels are described by certain functions, also ${\bf Ch}_W^{W^{\otimes2}}$ is closely associated with a particular set functions.

\begin{definition}
We denote by $C_S^0(G\times G)$ the set of continuous functions $f:G\times G\to\C$ with $f(e,e)=1$ such that, for any $n\in\N$ and all $g_1,\ldots,\,g_n,\,h_1,\ldots,\,h_n\in G$,
\begin{equation}\label{eq:gskewpositive}
\bigg(e^{-\frac{i}{2}\big(S(g_i,h_j)+S(h_i,g_j)\big)}f(g_i^{-1}g_j,h_i^{-1}h_j)\bigg)_{i,j=1}^n\geq0.
\end{equation}
\end{definition}

\begin{remark}\label{rem:gskewpositive}
We may make a couple of observations on the functions $f\in C_S^0(G\times G)$. Let us first show that any $f\in C_S^0(G\times G)$ is of positive type. Pick $f\in C_S^0(G\times G)$, $n\in\N$, and $g_i,\,h_i\in G$, $i=1,\ldots,\,n$. Make the substitution $g_i\to g_i^{-1}$, $h_i\to h_i^{-1}$, $i=1,\ldots,\,n$ in Equation \eqref{eq:gskewpositive}. After the change $i\leftrightarrow j$ in indexation, we obtain
$$
\bigg(e^{\frac{i}{2}\big(S(g_i,h_j)+S(h_i,g_j)\big)}f(g_i^{-1}g_j,h_i^{-1}h_j)\bigg)_{i,j=1}^n\geq0.
$$
When $(a_{i,j})_{i,j=1}^n$ and $(b_{i,j})_{i,j=1}^n$ are positive definite, also the entrywise product $(a_{i,j}b_{i,j})_{i,j=1}^n$ is positive. Thus Equation \eqref{eq:gskewpositive} and the equation above imply that $\big(f(g_i^{-1}g_j,h_i^{-1}h_j)\big)_{i,j=1}^n\geq0$. Thus, we see that $f$ is of positive type. We also immediately observe from Equation \eqref{eq:gskewpositive} that, for any $f\in C_S^0(G\times G)$, $f(\cdot,e),\,f(e,\cdot)\in C_+^0(G)$.

Let us also note that the function $f^S:G\times G\to\C$, $f^S(g,h)=e^{-\frac{i}{2}S(g,h)}$, $g,\,h\in G$, is not in $C_S^0(G\times G)$. This is a fact of great importance as we will see later on. To see this, pick $g,\,h\in G\setminus\{e\}$ such that $g\neq h$ and fix $g_1=g$, $g_2=e$, $h_1=h$, and $h_2=e$. Plugging these into the defining inequality \eqref{eq:gskewpositive} in the case $n=2$, the matrix on the LHS of said inequality becomes
$$
\left(\begin{array}{cc}
1&e^{-\frac{i}{2}S(g,h)}\\
e^{-\frac{i}{2}S(g,h)}&1
\end{array}\right)\not\geq0,
$$
showing that $f^S\notin C_S^0(G\times G)$. In fact, the above inequality also shows that $f^S$ is not positive either.
\end{remark}

\begin{theorem}\label{theor:gWcovjointchar}
For any $\Psi\in{\bf Ch}_W^{W^{\otimes2}}$ there is a unique $f_\Psi\in C_S^0(G\times G)$ such that
\begin{equation}\label{eq:gWcovjointchar}
\Psi\big(W(g)\otimes W(h)\big)=f_\Psi(g,h)W(gh),\qquad g,\,h\in G.
\end{equation}
Moreover the map ${\bf Ch}_W^{W^{\otimes2}}\ni\Psi\mapsto f_\Psi\in C_S^0(G\times G)$ is bijective.
\end{theorem}

\begin{proof}
Pick $\Psi\in{\bf Ch}_W^{W^{\otimes2}}$. Fix $g,\,g',\,h\in G$. Using Equation \eqref{eq:gWproperties} and the $(W,W^{\otimes2})$-covariance of $\Psi$, we have
\begin{eqnarray*}
&&W(g')W(gh)^*\Psi\big(W(g)\otimes W(h)\big)\\
&=&e^{\frac{i}{2}S(g',gh)}W(gh)^*W(g')\Psi\big(W(g)\otimes W(h)\big)W(g')^*W(g')\\
&=&e^{\frac{i}{2}S(g',gh)}W(gh)^*\Psi\big(W(g')W(g)W(g')^*\otimes W(g')W(h)W(g')^*\big)W(g')\\
&=&W(gh)^*\Psi\big(W(g)\otimes W(h)\big)W(g').
\end{eqnarray*}
Thus, for all $g,\,h\in G$, the operator $W(gh)^*\Psi\big(W(g)\otimes W(h)\big)$ commutes with the representation $W$ and is thus a scalar multiple of the identity operator. Hence, there is a function $f_\Psi:G\times G\to\C$ such that Equation \eqref{eq:gWcovjointchar} holds.

The claim now follows from Lemma \ref{lemma:CCRChan} upon noticing that $G\times G$ can be identified with the product phase space ${\bf X}^2\times\hat{\bf X}^2$ and defining the homomorphism $T:G\times G\to G$, $T(g,h)=gh$, $g,\,h\in G$.
\end{proof}

\begin{remark}
Recall the function $f^S:\,(g,h)\mapsto e^{-\frac{i}{2}S(g,h)}$ of Remark \ref{rem:gskewpositive}. We have already seen that $f^S\notin C_S^0(G\times G)$. Indeed, if this was not the case, there would be a channel $\Psi\in{\bf Ch}_W^{W^{\otimes2}}$,
$$
\Psi\big(W(g)\otimes W(h)\big)=f^S(g,h)W(gh)=W(g)W(h),\qquad g,\,h\in G.
$$
Since operators $W(g)\otimes W(h)$, $g,\,h\in G$, span an ultraweakly dense operator system, it follows that this channel would be the perfect cloner $A\otimes B\mapsto AB$, which is impossible.
\end{remark}

\subsection{Compatibility conditions for Weyl-covariant channels}\label{subsec:Wcompdilat}

We go on to give strict characterizations of compatibility of Weyl-covariant channels. Recalling Remark \ref{rem:0resource}, we can make a simple observation.

\begin{remark}\label{rem:translation}
Pick $\Phi_1,\,\Phi_2\in{\bf Ch}_W^W$ and let $\mu_1$ and $\mu_2$ be the Borel probability measures on $G$ such that
$$
\Phi_i(B)=\int_G W(g)^*BW(g)\,d\mu_i(g),\qquad B\in\mc L(\hil),\quad i=1,\,2.
$$
Equivalently, we may consider the continuous functions $f_i:G\to\C$ of positive type with $f_i(e)=1$, $i=1,\,2$, such that $\Phi_i\big(W(g)\big)=f_i(g)W(g)$, $g\in G$, $i=1,\,2$. These functions are Fourier transforms of the preceding measures, i.e.,
$$
f_i(g)=\int_G e^{-iS(g,h)}\,d\mu_i(h),\qquad g\in G,\quad i=1,\,2.
$$

Let $k_1,\,k_2\in G$. {\it The channels $\Phi_1$ and $\Phi_2$ are compatible if and only if the channels $\Phi_1^{k_1}$ and $\Phi_2^{k_2}$ associated, respectively, with the measures $\mu_1^{k_1}$ and $\mu_2^{k_2}$, where $\mu_i^{k_i}(X)=\mu_i(k_i^{-1}X)$ for all Borel sets $X\subseteq G$, $i=1,\,2$, are compatible.} The latter channels are equivalently associated with the functions $f_i^{k_i}$, $f_i^{k_i}(g)=e^{-iS(g,k_i)}f_i(g)$, $g\in G$, $i=1,\,2$. To see this, define $\Delta_1,\,\Delta_2\in{\bf Ch}_W^W$, $\Delta_i(B)=W(k_i)^*BW(k_i)$, $B\in\mc L(\hil)$, $i=1,\,2$. It is easily checked that $\Phi_i^{k_i}=\Phi_i\circ\Delta_i$, $i=1,\,2$. Thus, $\Phi_1^{k_1}$ and $\Phi_2^{k_2}$ are post-processings of $\Phi_1$ and $\Phi_2$. Since $\Delta_1$ and $\Delta_2$ are unitary, $\Phi_1$ and $\Phi_2$ are also post-processings of $\Phi_1^{k_1}$ and $\Phi_2^{k_2}$. The claim now follows from Remark \ref{rem:0resource}.
\end{remark}

The first compatibility characterization for Weyl-covariant channels follows directly from Proposition \ref{prop:covjoint} and Theorem \ref{theor:gWcovjointchar}.

\begin{corollary}
Let $\Phi_i\in{\bf Ch}_W^W$, $i=1,\,2$, and let $f_i:G\to\C$ be the corresponding continuous functions of positive type with $f_i(e)=1$ such that $\Phi_i\big(W(g)\big)=f_i(g)W(g)$ for all $g\in G$, $i=1,\,2$. The channels $\Phi_1$ and $\Phi_2$ are compatible if and only if there is a function $f\in C_S^0(G\times G)$ such that
$$
f_1(g)=f(g,e),\quad f_2(g)=f(e,g),\qquad g\in G.
$$
\end{corollary}

\begin{proof}
As already noted, $\Phi_1$ and $\Phi_2$ are compatible if and only if they are margins of a covariant channel $\Psi\in{\bf Ch}_W^{W^{\otimes2}}$. Fix such a covariant joint channel $\Psi\in{\bf Ch}_W^{W^{\otimes2}}$ and fix $f=f_\Psi$. It follows that, for all $g\in G$,
$$
f_1(g)W(g)=\Phi_1\big(W(g)\big)=\Psi\big(W(g)\otimes W(e)\big)=f(g,e)W(g),
$$
and similarly for $f_2$.
\end{proof}

From now on, fix a Haar measure $dg$ on $G$; the particular choice of measure is irrelevant. Whenever $\mu$ is a positive measure on the Borel $\sigma$-algebra of $G$ which is absolutely continuous with respect to $dg$, denote the Radon-Nikod\'ym derivative of $\mu$ with respect to $dg$ by $\rho_\mu$ and ${\rm supp}\,\mu:=\rho_\mu^{-1}\big((0,\infty]\big)$. Utilizing Proposition \ref{prop:covdilat}, we obtain another compatibility characterization. Note that particular restrictions are placed on the Borel probability measure associated with one of the channels. When this conditions is met, we may give a more direct compatibility condition of two Weyl-covariant channels. First, however, we give a lemma. From now on, all the `almost all' (a.a.) phrases are to be understood with respect to $dg$.

\begin{lemma}\label{lemma:zeta}
Suppose that $\mu$ is a probability measure on the Borel $\sigma$-algebra of $G$ which is absolutely continuous with respect to $dg$ and let $\mc L:=L^2(\mu)$. Denote by $\ovl W:G\to\mc U(\mc L)$ the representation,
$$
\big(\ovl W(g)f\big)(h)=e^{iS(g,h)}f(h),\qquad g\in G,\quad h\in G.
$$
For any $\tilde{\Phi}\in{\bf Ch}_{\ovl W}^W$ there is $\zeta:G\times G\to\C$ such that $\zeta(\cdot,h)$ is continuous and $\zeta(e,h)=1$ for a.a. $h\in G$, $\zeta(g,\cdot)$ is measurable and essentially bounded for all $g\in G$, for any $n\in\N$ and all $g_1,\ldots,\,g_n\in G$,
\begin{equation}\label{eq:Z4}
\bigg(e^{\frac{i}{2}S(g_i,g_j)}\zeta(g_i^{-1}g_j,g_ih)\bigg)_{i,j=1}^n\geq0,
\end{equation}
and
\begin{equation}\label{eq:tildePhichar}
\big(\tilde{\Phi}\big(W(g)\big)f\big)(h)=\zeta(g,h)f(gh)\sqrt{\frac{\rho_\mu(gh)}{\rho_\mu(h)}},\qquad g\in G,\quad h\in G,\quad f\in\mc L.
\end{equation}
Moreover, for any $\zeta:G\times G\to\C$ continuous in the first argument, measurable and essentially bounded in the second argument, with $\zeta(e,\cdot)=1$ a.e., and satisfying the condition of Equation \eqref{eq:Z4}, the Equation \eqref{eq:tildePhichar} defines a $\tilde{\Phi}\in{\bf Ch}_{\ovl{W}}^W$.
\end{lemma}

\begin{proof}
Pick $\tilde{\Phi}\in{\bf Ch}_{\ovl W}^W$. It is easy to see that
$$
\ovl W(g')\tilde{\Phi}\big(W(g)\big)=e^{iS(g',g)}\tilde{\Phi}\big(W(g)\big)\ovl W(g'),\qquad g,\,g'\in G.
$$
Denote by $L^2(G)$ the Hilbert space of (equivalence classes of) functions which are square integrable with respect to $dg$. One can define a decomposable isometry $J:\mc L\to L^2(G)$, $(Jf)(h)=\rho_\mu(h)^{1/2}f(h)$, $f\in\mc L$, $h\in G$. Clearly, for all $F\in L^2(G)$, $(J^*F)(h)=\rho_\mu(h)^{-1/2}F(h)$ whenever $h\in G$ and $(J^*F)(h)=0$ otherwise. Define $\tilde{\Phi}_J:=J\tilde{\Phi}(\cdot)J^*$ and the (canonical) diagonal representation $\delta:G\to\mc U\big(L^2(G)\big)$, $\big(\delta(g)F\big)(h)=e^{iS(g,h)}F(h)$, $g,\,h\in G$, $F\in L^2(G)$. It is easy to see that $J\ovl W(g)=\delta(g)J$ and $\delta(g')\tilde{\Phi}_J\big(W(g)\big)=e^{iS(g',g)}\tilde{\Phi}_J\big(W(g)\big)\delta(g')$, $g,\,g'\in G$. According to \cite[Lemma 1]{HaPe2017}, this is equivalent with the existence of a function $\zeta:G\times G\to\C$ such that $\zeta(g,\cdot)$ is measurable and essentially bounded for all $g\in G$ and
$$
\big(\tilde{\Phi}_J\big(W(g)\big)F\big)(h)=\zeta(g,h)F(gh),\qquad g\in G,\quad h\in G,\quad F\in L^2(G).
$$
Since the Radon-Nikod\'ym derivative of the measure $X\mapsto\mu(gX)$, $g\in G$, is $h\mapsto\rho_\mu(gh)$, the above is equivalent with Equation \eqref{eq:tildePhichar}. The unitality of $\tilde{\Phi}$ implies $\zeta(e,h)=1$ for a.a.\ $h\in G$.

We go on to the conditions the complete positivity of $\tilde{\Phi}$ sets on $\zeta$. We again fix $n,\,m\in\N$, $f_1,\ldots,\,f_n\in\mc L$, $\gamma_{i,r}\in\C$, and $g_{i,r}\in G$, $i=1,\ldots,\,n$, $r=1,\ldots,\,m$, and define $C_i=\sum_{r=1}^m\gamma_{i,r}W(g_{i,r})$. The complete positivity is, again, fully characterized by inequalities of the following kind (where the latter equality is obtained through a direct calculation using Equation \eqref{eq:tildePhichar} and the cocycle properties of the function $(g,h)\mapsto\rho_\mu(gh)$):
\begin{eqnarray*}
0&\leq&\sum_{i,j=1}^n\<f_i|\tilde{\Phi}(C_i^*C_j)f_j\>\\
&=&\sum_{i,j=1}^n\sum_{r,s=1}^m\ovl{\gamma_{i,r}}\gamma_{j,s}e^{\frac{i}{2}S(g_{i,r},g_{j,s})}\int_{G}\ovl{f_i(g_{i,r}h)}f_j(g_{j,s}h)\zeta(g_{i,r}^{-1}g_{j,s},g_{i,r}h)\sqrt{\rho_\mu(g_{i,r}h)\rho_\mu(g_{j,s}h)}\,dh\\
&=&\sum_{i,j=1}^n\sum_{r,s=1}^m\ovl{\gamma_{i,r}}\gamma_{j,s}e^{\frac{i}{2}S(g_{i,r},g_{j,s})}\int_G\ovl{(Jf_i)(g_{i,r}h)}(Jf_j)(g_{j,s}h)\zeta(g_{i,r}^{-1}g_{j,s},g_{i,r}h)\,dh.
\end{eqnarray*}
From this, after careful reading, one obtains the property of Equation \eqref{eq:Z4}.

Let us concentrate on the normality of $\tilde{\Phi}$. We probe this question again using the Fourier transform $\widehat{\tilde{\Phi}_*(S)}$ similarly as in the proof of Lemma \ref{lemma:CCRChan}. For simplicity, suppose $S=|f\>\<f|$, $f\in\mc L$. It follows, for every $g\in G$,
\begin{eqnarray*}
\widehat{\tilde{\Phi}_*(|f\>\<f|)}(g)=\<f|\tilde{\Phi}\big(W(g)\big)f\>&=&\int_{G}\ovl{f(h)}f(gh)\zeta(g,h)\sqrt{\rho_\mu(gh)\rho_\mu(h)}\,dh\\
&=&\int_G\ovl{(Jf)(h)}(Jf)(gh)\zeta(g,h)\,dh.
\end{eqnarray*}
As a function of $g$, the above expression has to be continuous and satisfy a condition analogous to the inequality \eqref{eq:BochnerPositive}. The latter condition is easily seen to reduce to \eqref{eq:Z4}. Using polarization, we see that the functions $g\mapsto\int_G\ovl{(Jf_1)(k)}(Jf_2)(gh)\zeta(g,h)\,dh=\<Jf_1|\tilde{\Phi}_J\big(W(g)\big)Jf_2\>$, $f_1,\,f_2\in\mc L$ also have to be continuous. Clearly we may replace $Jf_i$, $i=1,\,2$, with arbitrary functions $F_i\in L^2(G)$, $i=1,\,2$. Especially, for any compact $X\subset G$ and any $F\in C_0(G)$ (continuous and compactly supported functions on $G$), fixing $F_1=\chi_X$ and $F_2=F$, the map $g\mapsto\int_{X}\zeta(g,h)F(gh)\,dh$ is continuous. Since this holds for any compact $X$, we have that, for a.a. $h\in G$, the function $g\mapsto\zeta(g,h)F(gh)$ is continuous for any $F\in C_0(G)$. This implies that $\zeta(\cdot,h)$ is continuous for a.a. $h\in G$.

Thus, $\zeta$ satisfies the conditions of the claim. The reverse claim already follows from the proof thus far.
\end{proof}

We say that $\beta:G\times G\to\C$ is a {\it positive kernel} if it is measurable and, for any $n\in\N$ and all $g_1,\ldots,\,g_n\in G$,
$$
\big(\beta(g_i,g_j)\big)_{i,j=1}^n\geq0.
$$
Although the diagonal $\Delta=\{(g,g)\,|g\in G\}\subseteq G\times G$ is of zero measure, we may define the {\it diagonal values of $\beta$} as follows: Any positive kernel possesses a {\it Kolmogorov construction}, i.e.,\ a pair $(\mc M,\eta)$ where $\mc M$ is a Hilbert space and $\eta:G\to\mc M$ is a function such that $\beta(g,h)=\<\eta(g)|\eta(h)\>$ for all $g,\,h\in G$. Amongst these constructions there is an essentially unique minimal one where the vectors $\eta(g)$, $g\in G$, span a dense subspace of $\mc M$. The diagonal value at $g\in G$ is given by $\beta^\Delta(g):=\|\eta(g)\|^2$.

\begin{theorem}\label{theor:Wcovjointdilat}
Let $\Phi_1,\,\Phi_2\in{\bf Ch}_W^W$, and assume that the probability measure $\mu$ on the Borel $\sigma$-algebra of $G$ associated with $\Phi_1$ through
$$
\Phi_1(A)=\int_G\,W(g)^*AW(g)\,d\mu(g),\qquad A\in\mc L(\hil),
$$
is absolutely continuous with respect to $dg$. These channels are compatible if and only if there is a positive kernel $\beta:G\times G\to\C$ such that $\beta(\cdot,g)$ and $\beta(g,\cdot)$ are continuous for a.a.\ $g\in G$, $\beta^\Delta(g)=1$ for a.a.\ $g\in G$, and
\begin{equation}\label{eq:Phi2}
\Phi_2\big(W(g)\big)=\int_{G}\beta(h,gh)\sqrt{\rho_\mu(h)\rho_\mu(gh)}\,dh\,W(g),\qquad g\in G.
\end{equation}
\end{theorem}

\begin{proof}
Let $\mc L$ and $\ovl W$ be as in Lemma \ref{lemma:zeta}. Define the isometry $J:\hil\to\hil\otimes\mc L$, $(J\fii)(g)=\sqrt{\rho(g)}W(g)\fii$, $\fii\in\hil$, $g\in G$. It is simple to check that $(\mc L,J,\ovl W)$ is a covariant minimal Stinespring dilation for $\Phi_1$. According to Proposition \ref{prop:covdilat}, $\Phi_1$ and $\Phi_2$ are compatible if and only if there is $\tilde{\Phi}_2\in{\bf Ch}_{\ovl W}^W$ such that
\begin{equation}\label{eq:compdilatW}
\Phi_2(B)=J^*\big(\id_\hil\otimes\tilde{\Phi}_2(B)\big)J,\qquad B\in\mc L(\hil).
\end{equation}
Since channels $\tilde{\Phi}\in{\bf Ch}_{\ovl W}^W$ are, according to Lemma \ref{lemma:zeta}, in one-to-one correspondence with the functions $\zeta$ as detailed in the claim of Lemma \ref{lemma:zeta}, there must exist such $\zeta$ for $\tilde{\Phi}_2$.

We have that $\beta_k$, $\beta_k(g,h)=e^{\frac{i}{2}S(g,h)}\zeta(g^{-1}h,gk)$, $g,\,h\in G$, is a positive kernel with $\beta_k^\Delta=1$ a.e. for a.a.\ $k\in G$. Suppose that $k\in G$ is such that these conditions hold. Define $\beta=\beta_e$, whence it follows that $\beta(g,h)=e^{\frac{i}{2}\big(S(g,k)-S(h,k)\big)}\beta_k(gk^{-1},hk^{-1})$ implying that $\beta$ is a positive kernel with $\beta^\Delta(g)=1$ for a.a.\ $g\in G$. It follows immediately that $\beta(g,\cdot)$ is continuous for a.a.\ $g\in G$. Hence also $g\mapsto\beta(g,h)=\ovl{\beta(h,g)}$ is continuous for a.a.\ $h\in G$. Substituting the channel $\tilde{\Phi}_2=\tilde{\Phi}$ of \eqref{eq:tildePhichar} with $\zeta(g,h)=e^{\frac{i}{2}S(g,h)}\beta(h,gh)$, $g\in G$, $h\in G$, into \eqref{eq:compdilatW}, we obtain Equation \eqref{eq:Phi2}. On the other hand, $\zeta:G\times G\to\C$, $\zeta(g,h)=e^{\frac{i}{2}S(g,h)}\beta(h,gh)$, $g\in G$, $h\in G$, satisfies the conditions of the claim of Lemma \ref{lemma:zeta}.
\end{proof}

\begin{remark}\label{rem:alternative}
We may rephrase Theorem \ref{theor:Wcovjointdilat} in the following form: {\it Suppose that $\Phi_1\in{\bf Ch}_W^W$ is as in Theorem \ref{theor:Wcovjointdilat}. The continuous function $f_2:G\to\C$ of positive type with $f_2(e)=1$ associated with a channel $\Phi_2\in{\bf Ch}_W^W$ compatible with $\Phi_1$ such that $\Phi_2\big(W(g)\big)=f_2(g)W(g)$, $g\in G$, can be written in the form}
\begin{equation}\label{eq:alternative}
f_2(g)=\int_{G}\beta(h,gh)\sqrt{\rho_\mu(h)\rho_\mu(gh)}\,dh,\qquad g\in G,
\end{equation}
{\it where $\beta:G\times G\to\C$ is a positive kernel continuous in both arguments whose diagonal values are $1$ a.e.}

Suppose now that $\Phi_1,\,\Phi_2\in{\bf Ch}_W^W$ are associated with Borel probability measures measures $\mu_1$ and $\mu_2$ such that $\Phi_i(B)=\int_G W(g)^*BW(g)\,d\mu_i(g)$, $B\in\mc L(\hil)$, $i=1,\,2$. Assume that both $\mu_1$ and $\mu_2$ are absolutely continuous with respect to $dg$ and denote their Radon-Nikod\'ym derivates with respect to $dg$ by $\rho_1$ and, respectively, $\rho_2$. Also assume that $\hat{\rho_2}=f_2\in L^1(G)$ so that we may use the inverse Fourier transform to $f_2$ to obtain $\rho_2$. In this situation, $\Phi_1$ and $\Phi_2$ are compatible if and only if there is a positive kernel $\beta:G\times G\to\C$ such that $\beta(\cdot,g)$ and $\beta(g,\cdot)$ are continuous and $\beta^\Delta(g)=1$ for a.a.\ $g\in G$ and
\begin{equation}\label{eq:alternative2}
\rho_2(k)=\int_G\int_G e^{iS(g^{-1}h,k)}\beta(g,h)\sqrt{\rho_1(g)\rho_1(h)}\,dg\,dh,\qquad k\in G.
\end{equation}
The above equation is obtained from \eqref{eq:alternative} by taking the inverse Fourier transformation of $f_2$.

Especially if $\sqrt{\rho_1(\cdot)}\in L^1(G)$, any channel $\Phi_2\in{\bf Ch}_W^W$ compatible with $\Phi_1$ is associated with a measure $\mu_2$ which is absolutely continuous with respect to $dg$ and, consequently, compatibility is characterized by Equation \eqref{eq:alternative2}. Indeed, if this is the case, we may evaluate for the continuous function $f_2$ of positive type associated with $\Phi_2$
\begin{eqnarray*}
\int_G|f_2(g)|\,dg&=&\int_G\bigg|\int_{G}\beta(h,hg)\sqrt{\rho_1(h)\rho_1(gh)}\,dh\bigg|\,dg\\
&\leq&\int_G\int_{G}\underbrace{|\beta(h,gh)|}_{\leq1\ {\rm a.e.}}\sqrt{\rho_1(h)\rho_1(gh)}\,dh\,dg\leq\bigg(\int_{G}\sqrt{\rho_1(h)}\,dh\bigg)^2<\infty.
\end{eqnarray*}
Thus, the inverse Fourier transform can be applied to $f_2$ yielding the density function $\rho_2$ such that $\Phi_2(B)=\int_G \rho_2(g)W(g)^*BW(g)\,dg$, $B\in\mc L(\hil)$.

It should be noted that, in Proposition \ref{prop:covdilat}, the correspondence between channels $\tilde{\Phi}_2\in{\bf Ch}_{\ovl U}^{V_2}$ and $\Phi_2\in{\bf Ch}_U^{V_2}$ set up in \eqref{eq:covdilat} is many-to-one. Indeed, for any joint channel $\Psi\in{\bf Ch}_U^{V_{1,2}}$ for compatible $\Phi_i\in{\bf Ch}_U^{V_i}$, $i=1,\,2$, there is a unique $\tilde{\Phi}_2\in{\bf Ch}_{\ovl U}^{V_2}$ such that $\Psi(A\otimes B)=V^*\big(A\otimes\tilde{\Phi}_2(B)\big)V$, $A\in\mc L(\mc K_1)$, $B\in\mc L(\mc K_2)$, when we fix a covariant minimal Stinespring dilation $(\mc L,V,\ovl U)$ for $\Phi_1\in{\bf Ch}_U^{V_2}$. However, as any compatible pair of channels typically has infinitely many joint channels, there are typically many $\tilde{\Phi}_2$ such that Equation \ref{eq:covdilat} is satisfied. From this it follows that, for any $\Phi_2\in{\bf Ch}_W^W$ compatible with $\Phi_1\in{\bf Ch}_W^W$ satisfying the conditions of Theorem \ref{theor:Wcovjointdilat}, there are {\it a priori} several kernels $\beta$ satisfying Equation \eqref{eq:Phi2} or, equivalently, \eqref{eq:alternative}.
\end{remark}

\section{Physical phase spaces}\label{sec:PhysPhase}

Next we shall adapt the results of Section \ref{sec:gPhaseSpace} in a couple of physically motivated phase spaces. First we consider the case of a non-compact and continuous configuration space where ${\bf X}=\R^N=\hat{\bf X}$. After this, we briefly discuss the case of a finite configuration space ${\bf X}=\Z_d=\hat{\bf X}$, where $\Z_d:=\Z/(d\Z)$ for some $d\in\N$.

\subsection{Phase space $\R^N\times\R^N$}\label{sec:contphasespace}

We now take a closer look at the case of joinings Weyl-covariant channels of a spin-0 system in $N$-dimensional Euclidean configuration space ${\bf X}=\R^N$. The product $(\cdot|\cdot)$ is the natural scalar product in $\R^N$ and $S(\vec w_1,\vec w_2)=\vec q_1^T\vec p_2-\vec q_2^T\vec p_1=\vec w_1^T\Omega\vec w_2$ for $\vec w_i=(\vec q_i,\vec p_i)\in\R^N\times\R^N$, $i=1,\,2$, where
$$
\Omega=\left(\begin{array}{cc}
0&\id_N\\
-\id_N&0
\end{array}\right)
$$
is defined in the block form. We identify $(\R^N\times\R^N)\times(\R^N\times\R^N)=\R^{4N}$. The set $C_S^0(\R^{4N})$ consists of continuous functions $f:\R^{4N}\to\C$ such that, for any $n\in\N$ and all $\vec v_1,\ldots,\,\vec v_n,\,\vec w_1,\ldots,\,\vec w_n\in\R^{2N}$,
\begin{equation}\label{eq:continuousphasespace}
\bigg(e^{-\frac{i}{2}\big(\vec v_i^T\Omega\vec w_j+\vec w_i^T\Omega\vec v_j\big)}f(\vec v_j-\vec v_i,\vec w_j-\vec w_i)\bigg)_{i,j=1}^n\geq0.
\end{equation}

We will characterize the functions in $C_S^0(\R^{4N})$ corresponding to Gaussian channels in ${\bf Ch}_W^{W^{\otimes2}}$. Gaussian channels have been studied in conjunction with compatibility questions earlier, e.g.,\ in \cite{HeKiSchu2015} albeit in the case of joint measurability of observables. Recall that we may treat the representation $W_2:\R^{4N}\to\mc U(\hil\otimes\hil)$, $W_2(\vec{v},\vec{w})=W(\vec{v})\otimes W(\vec{w})$, $\vec{v},\,\vec{w}\in\R^{2N}$ as the Weyl representation of the product phase space $\R^{4N}$. The associated symplectic matrix is denoted $\Omega_2$,
$$
\Omega_2=\left(\begin{array}{cc}
\Omega&0\\
0&\Omega
\end{array}\right),
$$
and hence $W_2(\vec{z}_1)W_2(\vec{z}_2)=e^{-i\vec{z}_1^T\Omega_2\vec{z}_2}W_2(\vec{z}_2)W_2(\vec{z}_1)$, $\vec{z}_1,\,\vec{z_2}\in\R^{4N}$. Recall that a channel $\Psi\in{\bf Ch}(\hil,\hil\otimes\hil)$ is Gaussian, if and only if the Fourier transform $\widehat{\Psi_*(S)}$ is a Gaussian function for any positive $S\in\mc T(\hil)$ such that $\hat{S}$ is Gaussian. Equivalently \cite{GiCi2002}, there must be real-entry matrices $A\in\mc M_{2N\times 4N}(\R)$ and $B\in\mc M_{4N\times4N}(\R)$ satisfying
\begin{equation}\label{eq:compposGauss}
B+i\Omega_2-iA^T\Omega A\geq0
\end{equation}
and a $\vec{c}\in\R^{4N}$ (which can be chosen at random) such that $\Psi=\Psi_{A,B,\vec{c}}$,
\begin{equation}\label{eq:GaussCh}
\Psi_{A,B,\vec{c}}\big(W_2(\vec{z})\big)=e^{-\frac{1}{4}\vec{z}^TB\vec{z}-i\vec{c}^T\vec{z}}W(A\vec{z}),\qquad\vec{z}\in\R^{4N}.
\end{equation}

Define the linear map $\R^{2N}\ni\vec{w}\mapsto(\vec{w},\vec{w})\in\R^{4N}$ given by the matrix $J\in\mc M_{4N\times2N}(\R)$,
$$
J=\left(\begin{array}{c}
\id_{2N}\\
\id_{2N}
\end{array}\right).
$$
Note that we may write $W^{\otimes2}(\vec{w})=W_2(J\vec{w})$, $\vec{w}\in\R^{2N}$. Let $\Psi=\Psi_{A,B,\vec{c}}$ be as in equations \eqref{eq:compposGauss} and \eqref{eq:GaussCh} and assume additionally that $\Psi\in{\bf Ch}_W^{W^{\otimes2}}$. Hence, for all $\vec{w}\in\R^{2N}$ and $\vec{z}\in\R^{4N}$,
\begin{eqnarray*}
&&e^{-\frac{1}{4}\vec{z}^TB\vec{z}-i\vec{c}^T\vec{z}-i\vec{w}^T\Omega A\vec{z}}W(A\vec{z})=e^{-\frac{1}{4}\vec{z}^TB\vec{z}-i\vec{c}^T\vec{z}}W(\vec{w})W(A\vec{z})W(\vec{w})^*\\
&=&W(\vec{w})\Psi\big(W_2(\vec{z})\big)W(\vec{w})^*=\Psi\big(W_2(J\vec{w})W_2(\vec{z})W_2(J\vec{w})^*\big)\\
&=&e^{-i\vec{w}^TJ^T\Omega_2\vec{z}}\Psi\big(W_2(\vec{z})\big)=e^{-\frac{1}{4}\vec{z}^TB\vec{z}-i\vec{c}^T\vec{z}-i\vec{w}^TJ^T\Omega_2\vec{z}}W(A\vec{z}).
\end{eqnarray*}
Hence, $J^T\Omega_2=\Omega A$. Using the fact that $\Omega^2=-\id_{2N}$, we obtain through a simple calculation $A=-\Omega J^T\Omega_2=J^T$. Thus, a Gaussian channel $\Psi=\Psi_{A,B,\vec{c}}$ is $(W,W^{\otimes2})$-covariant if and only if $A=J^T=\big(\id_{2N}\ \id_{2N}\big)$. For such a covariant Gaussian channel, when we write $B=(B_{i,j})_{i,j=1}^2$ where $B_{i,j}\in\mc M_{2N\times2N}(\R)$, the condition \eqref{eq:compposGauss} takes the form
\begin{equation}\label{eq:covcompposGaus}
\left(\begin{array}{cc}
B_{1,1}&B_{1,2}-i\Omega\\
B_{2,1}-i\Omega&B_{2,2}
\end{array}\right)\geq0.
\end{equation}
The matrix on the LHS is self-adjoined if and only if $B_{1,1}$ and $B_{2,2}$ are symmetric and $B_{2,1}=B_{1,2}^T$; note that $\Omega$ is antisymmetric. Thus we have proven the following:

\begin{theorem}\label{theor:covGaussCh}
For a Gaussian channel $\Psi\in{\bf Ch}(\hil,\hil\otimes\hil)$, i.e.,\ $\Psi=\Psi_{A,B,\vec{c}}$ where $A$, $B$, and $\vec{c}$ satisfy equations \eqref{eq:compposGauss} and \eqref{eq:GaussCh}, we have $\Psi\in{\bf Ch}_W^{W^{\otimes2}}$ if and only if $A=\big(\id_{2N}\ \id_{2N}\big)$ and $B$ is a symmetric matrix such that
\begin{equation}\label{eq:covGaussCh}
\left(\begin{array}{cc}
B_{1,1}&B_{1,2}-i\Omega\\
B_{1,2}^T+i\Omega^T&B_{2,2}
\end{array}\right)\geq0
\end{equation}
when we write $B=(B_{i,j})_{i,j=1}^2$ in the block form. The function associated to $\Psi$ by Equation \eqref{eq:gWcovjointchar} is given by
$$
f_\Psi(\vec{z})=e^{-\frac{1}{4}\vec{z}^TB\vec{z}-i\vec{c}^T\vec{z}},\qquad\vec{z}\in\R^{4N}.
$$
\end{theorem}

\begin{remark}
Especially, choosing $B=\alpha\id_{4N}$ for any $\alpha\geq1$, inequality \eqref{eq:covGaussCh} is satisfied. Thus, the set $C_S^0(\R^{4N})$ contains a wide variety of Gaussian functions.

However, there are no modulus-1 functions in $C_S^0(\R^{4N})$, i.e.,\ functions $f:\R^{4N}\to\T$. Suppose that $f:\R^{4N}\to\T$ and make the counter assumption that $f\in C_S^0(\R^{4N})$. Since $f$ is also a function of positive type, also $\big(\ovl{f(\vec{v}_j-\vec{v}_i,\vec{w}_j-\vec{w}_i)}\big)_{i,j=1}^n\geq0$ for any $n\in\N$ and all $\vec{v}_1,\ldots,\,\vec{v}_n,\,\vec{w}_1,\ldots,\,\vec{w}_n\in\R^{2N}$. When we multiply the matrix here entry-wise with the matrix on the LHS of inequality \eqref{eq:continuousphasespace}, we find that, for any $n\in\N$ and all $\vec{v}_1,\ldots,\,\vec{v}_n,\,\vec{w}_1,\ldots,\,\vec{w}_n\in\R^{2N}$,
$$
\Big(e^{-\frac{i}{2}\big(\vec{v}_i^T\Omega\vec{w}_j+\vec{w}_i^T\Omega\vec{v}_j\big)}\Big)_{i,j=1}^n\geq0.
$$
Fix now $\vec{v},\,\vec{w}\in\R^{2N}\setminus\{0\}$ such that $\vec{v}\neq\vec{w}$. Note that $\vec{v}^T\Omega\vec{w}\neq0$. Fix $n=3$, $\vec{v}_i=\frac{1}{3}\vec{v}$, $i=1,\,2,\,3$, $\vec{w}_1=-\frac{1}{2}\vec{w}$, $\vec{w}_2=\frac{3}{2}\vec{w}$, and $\vec{w}_3=\frac{1}{2}\vec{w}$, and substitute these to the above equation. We denote the resulting matrix on the LHS of the equation above by $M$. Direct calculation shows
$$
\det{M}=2\Big(\cos{\Big({\small\frac{1}{2}}\vec{v}^T\Omega\vec{w}\Big)-1}\Big)<0
$$
implying that $M\not\geq0$.
\end{remark}

A channel $\Phi\in{\bf Ch}(\hil,\hil)$ is Gaussian if there are real-entry matrices $A_0,\,B_0\in\mc M_{2N\times2N}(\R)$ such that $B_0+i\Omega-iA_0^T\Omega A_0\geq0$ and a vector $\vec{c}_0\in\R^{2N}$ (which can be chosen at random) such that $\Phi=\Phi_{A_0,B_0,\vec{c}_0}$,
$$
\Phi\big(W(\vec{w})\big)=e^{-\frac{1}{4}\vec{w}^TB_0\vec{w}-i\vec{c}_0^T\vec{w}}W(A_0\vec{w}),\qquad\vec{w}\in\R^{2N}.
$$
Similarly as in the case of Gaussian channels $\Psi\in{\bf Ch}(\hil,\hil\otimes\hil)$, it can be easily checked that a Gaussian channel $\Phi=\Phi_{A_0,B_0,\vec{c}_0}\in{\bf Ch}(\hil,\hil)$ is $(W,W)$-covariant if and only if $A_0=\id_{2N}$ and $B_0\geq0$.

\begin{example}
Let us consider the compatibility of two covariant channels $\Phi_1,\,\Phi_2\in{\bf Ch}_W^W$ where $\Phi_1=\Phi_{\id_{2N},B,\vec{c}}$ is Gaussian from the perspective of Theorem \ref{theor:Wcovjointdilat} and Remark \ref{rem:alternative}. We assume that the measure associated with $\Phi_1$ is absolutely continuous with respect to the Lebesgue measure. This is easily seen to be equivalent with $B$ being of full rank; indeed the function $f_1:\R^{2N}\to\C$, $f_1(\vec{v})=e^{-\frac{1}{4}\vec{v}^TB\vec{v}-i\vec{c}^T\vec{v}}$, $\vec{v}\in\R^{2N}$, is in $L^1(\R^{2N})$ if and only if $B$ does not have the eigenvalue $0$. Taking the inverse Fourier transform of the preceding function, we obtain the density function $\rho_1:\R^{2N}\to\R$,
$$
\rho_1(\vec{w})=\frac{1}{\pi^N\sqrt{\det B}}e^{-(\Omega\vec{w}+\vec{c})^TB^{-1}(\Omega\vec{w}+\vec{c})},\qquad\vec{w}\in\R^{2N},
$$
of the measure $\mu_1$ associated with $\Phi_1$ with respect to the Lebesgue measure. Note that $\sqrt{\rho_1(\cdot)}\in L^1(\R^{2N})$ implying that any Weyl-covariant channel compatible with $\Phi_1$ is associated with a measure which is absolutely continuous with respect to the Lebesgue measure.

Suppose that $\Phi_2$ is associated with the continuous function $f_2$ of positive type with $f_2(0)=1$. According to Theorem \ref{theor:Wcovjointdilat}, $\Phi_1$ and $\Phi_2$ are compatible if and only if there is a positive kernel $\beta:\R^{2N}\times\R^{2N}\to\C$ such that $\beta(\cdot,\vec{w})$ and $\beta(\vec{w},\cdot)$ are continuous and $\beta^\Delta(\vec{w})=1$ for a.a.\ $\vec{w}\in\R^{2N}$ and, for all $\vec{v}\in\R^{2N}$,
\begin{eqnarray*}
f_2(\vec{v})&=&\frac{1}{\pi^N\sqrt{\det B}}\int_{\R^{2N}}\beta(\vec{w},\vec{v}+\vec{w})e^{-\frac{1}{2}(\Omega\vec{w}+\vec{c})^TB^{-1}(\Omega\vec{w}+\vec{c})-\frac{1}{2}(\Omega\vec{v}+\Omega\vec{w}+\vec{c})^TB^{-1}(\Omega\vec{v}+\Omega\vec{w}+\vec{c})}\,d\vec{w}\\
&=&\frac{e^{-\frac{1}{2}\vec{v}^T\Omega^TB^{-1}\Omega\vec{v}}}{\pi^N\sqrt{\det B}}\int_{\R^{2N}}\beta(\vec{w}-\vec{v}/2+\Omega\vec{c},\vec{w}+\vec{v}/2+\Omega\vec{c})e^{-\vec{w}^T\Omega^TB^{-1}\Omega\vec{w}}\,d\vec{w}.
\end{eqnarray*}
Replacing the kernel $\beta$ with the kernel $\beta':(\vec{v},\vec{w})\mapsto\beta(\vec{v}+\Omega\vec{c},\vec{w}+\Omega\vec{c})$, we see that we may omit the terms $\Omega\vec{c}$ in the equation above. If we further replace $\beta'$ with $\tilde{\beta}:(\vec{v},\vec{w})\mapsto\beta'(\Omega^TB^{1/2}\vec{v},\Omega^TB^{1/2}\vec{w})$, and make a change of variables, we find that the above $\Phi_1$ and $\Phi_2$ are compatible if and only if there is a positive kernel $\beta:\R^{2N}\times\R^{2N}\to\C$ such that $\beta(\cdot,\vec{v})$ and $\beta(\vec{v},\cdot)$ are continuous and $\beta^\Delta(\vec{v})=1$ for a.a.\ $\vec{v}\in\R^{2N}$ and
$$
f_2(\Omega^TB^{1/2}\vec{v})e^{\frac{1}{2}\|\vec{v}\|^2}=\frac{1}{\pi^N}\int_{\R^{2N}}\beta(\vec{w}-\vec{v}/2,\vec{w}+\vec{v}/2)e^{-\|\vec{w}\|^2}\,d\vec{w},\qquad\vec{v}\in\R^{2N}.
$$

For two Gaussian channels $\Phi_1=\Phi_{\id_{2N},B,\vec{c}}$ and $\Phi_2=\Phi_{\id_{2N},C,\vec{d}}$, where $B$ is of full rank, the above necessary and sufficient compatibility condition becomes
$$
e^{\frac{1}{2}\|\vec{v}\|^2-\frac{1}{4}\vec{v}^TB^{1/2}\Omega C\Omega^TB^{1/2}\vec{v}-i\vec{d}^T\Omega^TB^{1/2}\vec{v}}=\frac{1}{\pi^N}\int_{\R^{2N}}\beta(\vec{w}-\vec{v}/2,\vec{w}+\vec{v}/2)e^{-\|\vec{w}\|^2}\,d\vec{w},\qquad\vec{v}\in\R^{2N}.
$$
When we replace $\beta$ with $\beta'$, $\beta'(\vec{v},\vec{w})=e^{i\vec{d}^T\Omega^TB^{1/2}(\vec{v}-\vec{w})}\beta(\vec{v},\vec{w})$, $\vec{v},\,\vec{w}\in\R^{2N}$, we may omit the term $i\vec{d}^T\Omega^TB^{1/2}\vec{v}$ in the above equation. Thus, the above Gaussian $\Phi_1$ and $\Phi_2$ are compatible if and only if there is a positive kernel $\beta:\R^{2N}\times\R^{2N}\to\C$ such that $\beta(\cdot,\vec{v})$ and $\beta(\vec{v},\cdot)$ are continuous and $\beta^\Delta(\vec{v})=1$ for a.a.\ $\vec{v}\in\R^{2N}$ and
\begin{equation}\label{eq:CompGaussCond}
e^{\frac{1}{2}\|\vec{v}\|^2-\frac{1}{4}\vec{v}^TB^{1/2}\Omega C\Omega^TB^{1/2}\vec{v}}=\frac{1}{\pi^N}\int_{\R^{2N}}\beta(\vec{w}-\vec{v}/2,\vec{w}+\vec{v}/2)e^{-\|\vec{w}\|^2}\,d\vec{w},\qquad\vec{v}\in\R^{2N}.
\end{equation}
For the compatibility of the above $\Phi_1$ and $\Phi_2$ it is necessary that $\id-\frac{1}{2}B^{1/2}\Omega C\Omega^TB^{1/2}\leq0$ or, equivalently,
\begin{equation}\label{eq:NecCompGaussCond}
\Omega C\Omega^T\geq2B^{-1}.
\end{equation}
Indeed, if this was not the case, there would exist $\vec{v}_+\in\R^{2N}$ such that, upon substituting $\vec{v}=\vec{v}_+$, the exponent on the LHS of Equation \eqref{eq:CompGaussCond} is positive and, consequently, the LHS is greater than 1 whereas, since $\beta(\vec{v},\vec{w})\leq1$ for a.a.\ $\vec{v},\,\vec{w}\in\R^{2N}$, the RHS is bounded from the above by 1 for every $\vec{v}\in\R^{2N}$. Thus, in order to guarantee that two Gaussian channels are incompatible, it suffices to ensure that the inequality \eqref{eq:NecCompGaussCond} is not satisfied.

There are also simple sufficient compatibility conditions for covariant Gaussian channels: Let $\Phi_i\in{\bf Ch}_W^W$, $i=1,\,2$, be Gaussian channels, i.e.,\ there are positive $B_{i,i}\in\mc M_{2N\times2N}(\R)$ and vectors $\vec{c}_i\in\R^{2N}$ such that $\Phi_i=\Phi_{\id_{2N},B_{i,i},\vec{c}_i}$, $i=1,\,2$. We do not have to assume that $B_{1,1}$ or $B_{2,2}$ be of full rank. The channels $\Phi_1$ and $\Phi_2$ are compatible if there is $B_{1,2}\in\mc M_{2N\times2N}(\R)$ such that the inequality \eqref{eq:covGaussCh} holds. To see this, suppose that there is a real matrix $B_{1,2}$ such that the inequality \eqref{eq:covGaussCh} holds. Denote $B=(B_{i,j})_{i,j=1}^2$ and $\vec{c}=(\vec{c}_1,\vec{c}_2)\in\R^{4N}$. According to Theorem \ref{theor:covGaussCh}, we may define the Gaussian channel $\Psi:=\Psi_{J^T,B,\vec{c}}\in{\bf Ch}_W^{W^{\otimes2}}$. Considering images $\Psi\big(W_2(\vec{w},0)\big)=\Psi_{(1)}\big(W(\vec{w})\big)$ and $\Psi\big(W_2(0,\vec{w})\big)=\Psi_{(2)}\big(W(\vec{w})\big)$, $\vec{w}\in\R^{2N}$, one easily finds that $\Psi_{(i)}=\Phi_i$, $i=1,\,2$.
\end{example}

\subsection{Phase space $\Z_d\times\Z_d$}\label{subsec:finitephase}

When the configuration space is the finite $\Z_d$ with $d\in\N$, the scalar product is $\Z_d\times\Z_d\ni(k,l)\mapsto(k|l)=\frac{2\pi}{d}kl\in\R$ and $S(\vec{m}_1,\vec{m}_2)=\frac{2\pi}{d}(k_1l_2-k_2l_1)$, $\vec{m}_i=(k_i,l_i)\in\Z_d^2$, $i=1,\,2$. The set $C_S^0(\Z_d^2\times\Z_d^2)$ consists of functions $f:\Z_d^2\times\Z_d^2\to\C$ with $f(0,0)=1$ such that the multi-index matrix inequality
$$
\bigg(e^{\frac{i}{2}\big(S(\vec{m},\vec{s})+S(\vec{r},\vec{n})\big)}f(\vec{m}-\vec{n},\vec{r}-\vec{s})\bigg)_{(\vec{m},\vec{r}),\,(\vec{n},\vec{s})\in\Z_d^2}\geq0
$$
holds.

Naturally, the function $f_0:\Z_d\times\Z_d\to\C$,
$$
f_0(\vec{m},\vec{n})=\left\{\begin{array}{ll}
1,&\vec{m}=0=\vec{n},\\
0&{\rm otherwise},
\end{array}\right.
$$
is in $C_S^0(\Z_d^2\times\Z_d^2)$. One easily checks that the corresponding Weyl-covariant joint channel $W\in{\bf Ch}_W^{W^{\otimes2}}$ is the completely depolarizing channel,
\begin{equation}\label{eq:depolarizing}
\Psi_0(C)=\frac{1}{d}\tr{C}\id_{\hil},\qquad C\in\mc L(\hil\otimes\hil).
\end{equation}
Naturally, the Hilbert space $\hil$ is now $\ell_{\Z_d}^2\simeq\C^d$.

Let us investigate the consequences of Theorem \ref{theor:Wcovjointdilat} in the finite phase space case. Note that the condition for $\Phi_1$ in the said theorem is now automatically satisfied. We say that a function $\beta:\Z_d^2\times\Z_d^2\to\C$ is a {\it positive kernel}, if $\big(\beta(\vec{m},\vec{n})\big)_{\vec{m},\vec{n}\in\Z_d^2}\geq0$. The following is an immediate consequence of the discussion in Remark \ref{rem:alternative}.

\begin{corollary}\label{prop:finitecomp}
Let $\Phi_1,\,\Phi_2\in{\bf Ch}_W^W$ be associated with probability vectors $p_1,\,p_2:\Z_d^2\to[0,1]$,
$$
\Phi_i(B)=\sum_{\vec{m}\in\Z_d^2}p_i(\vec{m})W(\vec{m})^*BW(\vec{m}),\qquad B\in\mc L(\hil),\quad i=1,\,2.
$$
These channels are compatible if and only if there is a positive kernel $\beta:\Z_d^2\times\Z_d^2\to\C$ with $\beta(\vec{m},\vec{m})=1$ for all $\vec{m}\in\Z_d^2$ such that
\begin{equation}\label{eq:finitecomp}
p_2(\vec{r})=\frac{1}{d^2}\sum_{\vec{m},\vec{n}\in\Z_d^2}e^{iS(\vec{m}-\vec{n},\vec{r})}\beta(\vec{m},\vec{n})\sqrt{p_1(\vec{m})p_1(\vec{n})},\qquad\vec{r}\in\Z_d^2.
\end{equation}
\end{corollary}

Corollary \ref{prop:finitecomp} gives us a recipe of finding all the channels $\Phi_2\in{\bf Ch}_W^W$ compatible with a fixed $\Phi_1\in{\bf Ch}_W^W$. Let us first take a look at two simple cases. Consider first the case where $\Phi_1(A)=W(\vec{m}_0)^*AW(\vec{m}_0)$ for all $A\in\mc L(\hil)$ with some fixed $\vec{m}_0\in\Z_d^2$. Now $p_1(\vec{m}_0)=1$ and $p_1(\vec{m})=0$ whenever $\vec{m}\neq\vec{m}_0$. When $\beta:\Z_d^2\times\Z_d^2\to\C$ is a positive kernel with $\beta(\vec{m},\vec{m})=1$ for all $\vec{m}\in\Z_d^2$, it follows that the channel associated with $\beta$ according to Corollary \ref{prop:finitecomp} is determined by the probability vector $p_2:\Z_d^2\to[0,1]$,
$$
p_2(\vec{r})=\frac{1}{d^2}\sum_{\vec{m},\vec{n}\in\Z_d^2}e^{iS(\vec{m}-\vec{n},\vec{r})}\beta(\vec{m},\vec{n})\sqrt{p_1(\vec{m})p_1(\vec{n})}=\frac{1}{d^2}\beta(\vec{m}_0,\vec{m}_0)=\frac{1}{d^2},\quad\vec{r}\in\Z_d^2.
$$
Hence, the only channel in ${\bf Ch}_W^W$ which is compatible with $\Phi_1$ is the depolarizing channel $\Psi_0$ of Equation \eqref{eq:depolarizing}. This fact is, of course, well known; whenever $\Phi$ is a unitary channel, the only channels compatible with $\Phi$ are of the form $B\mapsto\tr{\sigma B}\id$ for some positive trace-1 operator $\sigma$ and the only one of these within ${\bf Ch}_W^W$ is the depolarizing channel corresponding to $\sigma=d^{-1}\id_\hil$.

On the other hand, let $\Phi_1\in{\bf Ch}_W^W$ be a depolarizing channel. This is easily seen to correspond to $p_1(\vec{m})=d^{-2}$ for all $\vec{m}\in\Z_d^2$. Pick any probability vector $p_2:\Z_d^2\to[0,1]$ and define the positive kernel $\beta:\Z_d^2\times\Z_d^2\to\C$, $\beta(\vec{m},\vec{n})=\sum_{\vec{s}\in\Z_d^2}p_2(\vec{s})e^{iS(\vec{n}-\vec{m},\vec{s})}$, $\vec{m},\,\vec{n}\in\Z_d^2$. It is easily checked that $\beta(\vec{m},\vec{m})=1$ for all $\vec{m}\in\Z_d^2$ and
$$
\frac{1}{d^2}\sum_{\vec{m},\vec{n}\in\Z_d^2}e^{iS(\vec{m}-\vec{n},\vec{r})}\beta(\vec{m},\vec{n})\sqrt{p_1(\vec{m})p_1(\vec{n})}=p_2(\vec{r}),\quad\vec{r}\in\Z_d^2.
$$
Thus, $\Phi_1$ is compatible with any channel $\Phi_2\in{\bf Ch}_W^W$. Indeed, it is well known that $\Phi_1$ is compatible with any channel, not only with covariant channels.

\begin{example}\label{ex:Zd}
Let us consider an example where we mix white noise in the form of the completely depolarizing channel to two unitary Weyl-covariant channels and determine the conditions the mixing parameters have to satisfy so that the approximate unitary channels are compatible. Our target unitary channels are $\Phi_1$ and $\Phi_2$, $\Phi_i(B)=W(\vec{m}_i)^*BW(\vec{m}_i)$, $B\in\mc L(\hil)$, $i=1,\,2$, where $\vec{m}_1,\,\vec{m}_2\in\Z_d^2$ are some fixed phase space points. These channels correspond to the point mass probability vectors $\delta_{\vec{m}_i}$, $i=1,\,2$, i.e.,\ $\delta_{\vec{m}_i}(\vec{m}_i)=1$ and $\delta_{\vec{m}_i}(\vec{m})=0$ otherwise, $i=1,\,2$. The white noise is represented by the channel $\Phi^0:B\mapsto d^{-1}\tr{B}\id_\hil$ corresponding to the uniform probability vector $p^0$, $p^0(\vec{m})=d^{-2}$ for all $\vec{m}\in\Z_d^2$.
Define the channels $\Phi_i^s\in{\bf Ch}_W^W$, $0\leq s\leq1$, $i=1,\,2$, $\Phi_i^s=(1-s)\Phi_i+s\Phi^0$, $i=1,\,2$. These correspond to probability vectors $p_i^s=(1-s)\delta_{\vec{m}_i}+sp^0$, $0\leq s\leq1$, $i=1,\,2$.

Let us look at the compatibility conditions of channels $\Phi_1^s$ and $\Phi_2^t$, $s,\,t\in[0,1]$. According to Remark \ref{rem:translation}, we may simplify this task: we may simply assume that $\vec{m}_1=0=\vec{m}_2$ and still obtain the compatibility conditions of the more general situation described above. Thus, we are looking at the conditions $s$ and $t$ have to satisfy so that the channels $\Phi^s=(1-s){\rm id}+s\Phi^0$ and $\Phi^t=(1-t){\rm id}+t\Phi^0$ are compatible. These channels correspond to the probability vectors $p^s=(1-s)\delta_0+sp^0$ and $p^t=(1-t)\delta_0+tp^0$.

Direct substitution of $p^s$ and $p^t$ into Equation \eqref{eq:finitecomp} yields for all $\vec{r}\in\Z_d^2$
\begin{eqnarray}\label{eq:sij}
(1-t)\delta_{\vec{r},0}+\frac{t}{d^2}&=&\frac{1}{d^2}\bigg(1-\frac{d^2-1}{d^2}s\bigg)+2\frac{\sqrt{s}}{d^3}\sqrt{1-\frac{d^2-1}{d^2}s}\sum_{\vec{m}\in\Z_d^2\setminus\{0\}}{\rm Re}\big(e^{iS(\vec{m},\vec{r})}\beta(\vec{m},0)\big)\nonumber\\
&&+\frac{s}{d^4}\sum_{\vec{m},\vec{n}\in\Z_d^2\setminus\{0\}}e^{iS(\vec{m}-\vec{n},\vec{r})}\beta(\vec{m},\vec{n}),
\end{eqnarray}
where $\delta_{\vec{m},\vec{r}}$ stands for the Kronecker symbol, i.e.,\ $\delta_{\vec{m},\vec{n}}=1$ if and only if $\vec{m}=\vec{n}$ and otherwise $\delta_{\vec{m},\vec{n}}=0$. To find the boundary of the region of those $(s,t)\in[0,1]^2$ such that $\Phi^s$ and $\Phi^t$ are compatible, we look at the minimum value of $t$ for each $s$ such that $\Phi^s$ and $\Phi^t$ are compatible. We denote this value by $t_{\rm min}(s)$; this is found, e.g.,\ by substituting $\vec{r}=0$ in Equation \eqref{eq:sij}, setting $1-t_{\rm min}(s)+t_{\rm min}/d^2$ on the LHS and choosing $\beta$ so that the RHS is maximized. This is due to the fact that the entry $p^t(0)$ is the largest of $p^t$ and by maximizing this, the noise terms $p^t(\vec{r})=t/d^2$, $\vec{r}\in\Z_d^2\setminus\{0\}$, are minimized.

We have $|\beta(\vec{m},\vec{n})|\leq1$ for all $\vec{m},\,\vec{n}\in\Z_d^2$, and it is immediately seen that the above maximization task is solved by setting $\beta(\vec{m},\vec{n})=1$ for all $\vec{m},\,\vec{n}\in\Z_d^2$. Substituting this into Equation \eqref{eq:sij} where the LHS is $1-t_{\rm min}(s)+t_{\rm min}/d^2$, through direct calculation we obtain $t_{\rm min}(s)=\big(\sqrt{1-(1-1/d^2)s}-\sqrt{s}/d\big)^2$. This means that the above channels are compatible if and only if $\sqrt{t}\geq\sqrt{1-d^{-2}(d^2-1)s}-\sqrt{s}/d$. Moving the last term on the RHS of this inequality to the LHS and squaring both sides of the resulting inequality we arrive at the following: the channels $\Phi^s$ and $\Phi^t$ (or equivalently the channels $\Phi_1^s$ and $\Phi_2^t$ of the beginning of this example) are compatible if and only if
\begin{equation}\label{eq:finiteex}
s+\frac{2}{d}\sqrt{st}+t\geq1.
\end{equation}

\begin{figure}
\includegraphics[scale=1.3]{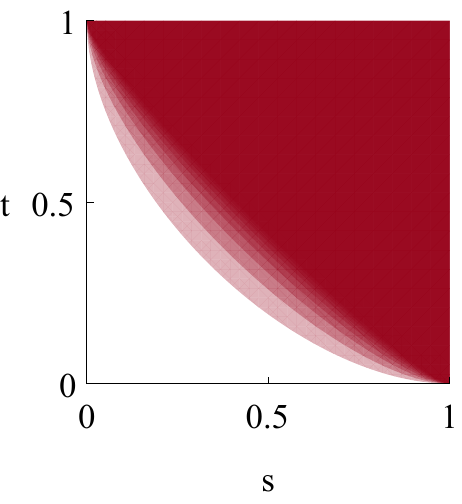}
\caption{\label{fig:F1} The area characterized by the inequality \eqref{eq:finiteex} of those noise parameters $(s,t)\in[0,1]\times[0,1]$ such that $\Phi_1^s$ and $\Phi_2^t$ are compatible is plotted here with varying dimension $d=2,\ldots,\,10$ (size of configuration space). Each darker shade indicates a rise in $d$ by $1$, i.e.,\ the parameter values from the whole coloured area are the ones with which $\Phi_1^s$ and $\Phi_2^t$ are compatible in $d=2$, those from the darker hue correspond to compatibility in $d=3$, and so forth. We see that, as $d\to\infty$, the area of compatibility grows closer to just the upper right triangle. The white area is the area where $\Phi_1^s$ and $\Phi_2^t$ are incompatible in any dimension.}
\end{figure}

We immediately notice that, as the dimension $d$ increases, the inequality \eqref{eq:finiteex} grows ever closer to $s+t\geq1$; this is also demonstrated in Figure \ref{fig:F1}. Pairs of unitary channels are the pairs whose incompatibility is the most resistant under added noise and, as dimension increases, we near the dimension-independent ultimate noise tolerance, since if the noise parameters satisfy $s+t\geq1$, the channels $\Phi^s$ and $\Phi^t$ are guaranteed to be compatible \cite{BuHeSchuSte2013,Haapasalo2015}. In the symmetric case where $s=t$, we see that the channels $\Phi_1^t$ and $\Phi_2^t$ are compatible if and only if $t\geq d/\big(2(d+1)\big)$. This result is in line with earlier findings \cite{Haapasalo2015}.
\end{example}

\begin{example}
Let us consider the case of the simple phase space $\Z_2\times\Z_2$. Naturally, the Hilbert space associated with the phase space is two dimensional. Let us fix a basis in $\C^2$ and define the Pauli matrices:
$$
\sigma_0=\id_2,\quad
\sigma_1=\left(\begin{array}{cc}
0&1\\
1&0
\end{array}\right),\quad
\sigma_2=\left(\begin{array}{cc}
0&-i\\
i&0
\end{array}\right),\quad
\sigma_3=\left(\begin{array}{cc}
1&0\\
0&-1
\end{array}\right).
$$
When we define $W:\Z_2^2\to\mc L(\C^2)$, $W(0,0)=\sigma_0$, $W(0,1)=\sigma_1$, $W(1,0)=\sigma_2$, and $W(1,1)=\sigma_3$, we notice that this map satisfies the CCR-relations \eqref{eq:gWproperties} meaning that $W$ is the Weyl representation associated with the phase space $\Z_2\times\Z_2$.

Any channel $\Phi\in{\bf Ch}_W^W$ is determined by a probability vector $p:\Z_2^2\to[0,1]$ through $\Phi(B)=\sum_{k,l\in\Z_2}p(k,l)W(k,l)^*BW(k,l)$, $B\in\mc L(\C^2)$. According to Corollary \ref{prop:finitecomp}, two channels $\Phi_1,\,\Phi_2\in{\bf Ch}_W^W$ associated, respectively, with probability vectors $p_1,\,p_2:\Z_2^2\to[0,1]$ are compatible if and only if there is a positive matrix
$$B=\big(\beta(\vec{m},\vec{n})\big)_{\vec{m},\vec{n}\in\Z_2^2}=\left(\begin{array}{cccc}
1&a&b&d\\
\ovl a&1&c&e\\
\ovl b&\ovl c&1&f\\
\ovl d&\ovl e&\ovl f&1
\end{array}\right)
$$
such that Equation \eqref{eq:finitecomp} is satisfied with $d=2$.

We immediately notice that the phase terms $e^{iS(\vec{m}-\vec{n},\vec{r})}$ in \eqref{eq:finitecomp} are all real, and taking the complex conjugate on both sides of this equation and, summing the resulting (true) equation with the original one, we find that we are free to assume that $B\in\mc M_{4\times 4}(\R)$, i.e.,\ $\beta(\vec{m},\vec{n})\in\R$ for all $\vec{m},\,\vec{n}\in\Z_2^2$. We are free to set $a=b=d=0$ if $p_1(0,0)=0$, $c=e=a=0$ if $p_1(0,1)=0$, $f=b=c=0$ if $p_1(1,0)=0$, and $d=e=f=0$ if $p_1(1,1)=0$; the resulting matrix is positive whenever the original $B$ is positive and satisfies Equation \eqref{eq:finitecomp} similarly as $B$. Also note that $B\geq0$ if and only if $\tilde{B}\geq0$ where
\begin{eqnarray*}
\tilde{B}&=&\left(\begin{array}{cccc}
p_1(0,0)&\tilde{a}&\tilde{b}&\tilde{d}\\
\tilde{a}&p_1(0,1)&\tilde{c}&\tilde{e}\\
\tilde{b}&\tilde{c}&p_1(1,0)&\tilde{f}\\
\tilde{d}&\tilde{e}&\tilde{f}&p_1(1,1)\\
\end{array}\right)\\
&:=&\left(\begin{array}{cccc}
p_1(0,0)&a\sqrt{p_1(0,0)p_1(0,1)}&b\sqrt{p_1(0,0)p_1(1,0)}&d\sqrt{p_1(0,0)p_1(1,1)}\\
a\sqrt{p_1(0,0)p_1(0,1)}&p_1(0,1)&c\sqrt{p_1(0,1)p_1(1,0)}&e\sqrt{p_1(0,1)p_1(1,1)}\\
b\sqrt{p_1(0,0)p_1(1,0)}&c\sqrt{p_1(0,1)p_1(1,0)}&p_1(1,0)&f\sqrt{p_1(1,0)p_1(1,1)}\\
d\sqrt{p_1(0,0)p_1(1,1)}&e\sqrt{p_1(0,1)p_1(1,1)}&f\sqrt{p_1(1,0)p_1(1,1)}&p_1(1,1)
\end{array}\right)
\end{eqnarray*}
is the entrywise product of $B$ and $\big(\sqrt{p_1(\vec{m})p_1(\vec{n})}\big)_{\vec{m},\vec{n}\in\Z_2^2}$ (which is positive). We obtain the original $B$ from $\tilde{B}$ by taking the entrywise product of $\tilde{B}$ and the positive $\big(c_{\vec{m},\vec{n}}\big)_{\vec{m},\vec{n}\in\Z_2^2}$ where $c_{\vec{m},\vec{n}}=\big(p_1(\vec{m})p_1(\vec{n})\big)^{-1/2}$ whenever $p_1(\vec{m})\neq 0\neq p_1(\vec{n})$, and $c_{\vec{m},\vec{n}}=0$ otherwise; note the above freedom in nullifying certain entries of $B$ depending on the vanishing of $p_1$. Making these assumptions, we find, using \eqref{eq:finitecomp},
$$
\left\lbrace\begin{array}{rcl}
p_2(0,0)+p_2(0,1)&=&\frac{1}{2}+\tilde{a}+\tilde{f},\\
p_2(0,0)+p_2(1,0)&=&\frac{1}{2}+\tilde{b}+\tilde{e},\\
p_2(0,0)+p_2(1,1)&=&\frac{1}{2}+\tilde{d}+\tilde{c}.
\end{array}\right.
$$

From this we find that $\Phi_1$ and $\Phi_2$ are compatible if and only if there are $x,\,y,\,z\in\R$ such that
\begin{equation}\label{eq:qubitcond1}
\left(\begin{array}{cccc}
p_1(0,0)&x&y&d(z)\\
x&p_1(0,1)&z&e(y)\\
y&z&p_1(1,0)&f(x)\\
d(z)&e(y)&f(x)&p_1(1,1)
\end{array}\right)\geq0,
\end{equation}
where
$$
\left\lbrace\begin{array}{rcl}
d(z)&=&\frac{1}{2}\big(p_2(0,0)-p_2(0,1)-p_2(1,0)+p_2(1,1)\big)-z,\\
e(y)&=&\frac{1}{2}\big(p_2(0,0)-p_2(0,1)+p_2(1,0)-p_2(1,1)\big)-y,\\
f(x)&=&\frac{1}{2}\big(p_2(0,0)+p_2(0,1)-p_2(1,0)-p_2(1,1)\big)-x.
\end{array}\right.
$$
This is immediately seen as the existence of $x,\,y,\,z\in\R$ such that \eqref{eq:qubitcond1} is satisfied is equivalent with the existence of a positive $B=\big(\beta(\vec{m},\vec{n})\big)_{\vec{m},\vec{n}\in\Z_2^2}$ such that Equation \eqref{eq:finitecomp} is satisfied.

\end{example}

\section{Notes on the multipartite case}\label{sec:multipartite}

All the results of this paper except for Proposition \ref{prop:covdilat}, Theorem \ref{theor:Wcovjointdilat}, Corollary \ref{prop:finitecomp}, and their corollaries can be generalized in a straight-forward manner to the multipartite case.

\begin{definition}
Fix $m\in\N$ and Hilbert spaces $\hil$ and $\mc K_i$, $i=1,\ldots,\,m$. Suppose that $\Psi\in{\bf Ch}(\hil,\mc K_1\otimes\cdots\otimes\mc K_m)$ and define the {\it $i$'th margin of $\Psi$}, $i=1,\ldots,\,m$,
$$
\Psi_{(i)}(B)=\Psi(\id_{\mc K_1}\otimes\cdots\otimes\id_{\mc K_{i-1}}\otimes B\otimes\id_{\mc K_{i+1}}\otimes\cdots\otimes\id_{\mc K_m}),\qquad B\in\mc L(\mc K_i).
$$
Channels $\Phi_i\in{\bf Ch}(\hil,\mc K_i)$, $i=1,\ldots,\,m$ are {\it compatible} if there is a {\it joint channel} $\Psi\in{\bf Ch}(\hil,\mc K_1\otimes\cdots\otimes\mc K_m)$ such that $\Phi_i=\Psi_{(i)}$, $i=1,\ldots,\,m$.
\end{definition}

Fix Hilbert spaces $\hil$ and $\mc K_i$, $i=1,\ldots,\,m$, a group $G$, and projective unitary representations $U:G\to\mc U(\hil)$ and $V_i:G\to\mc U(\mc K_i)$, $i=1,\ldots,\,m$. Also define $V:G\to\mc U(\mc K_1\otimes\cdots\mc K_m)$,
$$
V(g)=V_1(g)\otimes\cdots\otimes V_m(g),\qquad g\in G.
$$
With obvious modifications, the proofs of Lemma \ref{lemma:jointnormal} and Proposition \ref{prop:covjoint} can be adapted to the multipartite case. Hence, we have:

\begin{proposition}
Suppose that $\hil$ and $\mc K_i$, $i=1,\ldots,\,m$, are separable Hilbert spaces, $G$ is an amenable locally compact group, and $U:G\to\mc U(\hil)$ and $V_i:G\to\mc U(\mc K_i)$, $i=1,\ldots,\,m$, are strongly continuous projective unitary representations. Any $m$-tuple $\Phi_i\in{\bf Ch}_U^{V_i}$, $i=1,\,2$, of compatible covariant channels has a joint channel $\ovl\Psi\in{\bf Ch}_U^V$, where $V$ is defined as above.
\end{proposition}

Consider now the general phase space $G={\bf X}\times\hat{\bf X}$ and the Weyl representation $W$ of Section \ref{sec:gPhaseSpace}. Again, the compatibility properties hinge on the set ${\bf Ch}_W^{W^{\otimes m}}$, where $W^{\otimes m}:G\to\mc U(\hil^{\otimes m})$,
$$
W^{\otimes m}(g)=\underbrace{W(g)\otimes\cdots\otimes W(g)}_{m\ {\rm copies}}.
$$

\begin{definition}
Denote by $C_S^0(G^N)$ the set of continuous functions $f:G^N\to\C$ with $f(e,\ldots,e)=1$ such that, for any $n\in\N$ and $g_{k,i}\in G$, $k=1,\ldots,\,m$, $i=1,\ldots,\,n$,
$$
\left(f(g_{1,i}^{-1}g_{1,j},\ldots,g_{m,i}^{-1}g_{m,j})\prod_{\underset{k\neq l}{k,l=1}}^m e^{-\frac{i}{2}S(g_{k,i},g_{l,j})}\right)_{i,j=1}^n\geq0.
$$
\end{definition}

Again, the proofs of Proposition \ref{prop:Wirr} and Theorem \ref{theor:gWcovjointchar} are easily generalized and we obtain:

\begin{theorem}
For any $\Psi\in{\bf Ch}_W^{W^{\otimes m}}$ there is a unique $f_\Psi\in C_S^0(G^N)$ such that
$$
\Psi\big(W(g_1)\otimes\cdots\otimes W(g_m)\big)=f_\Psi(g_1,\ldots,g_m)W(g_1\cdots g_m),\qquad g_1,\ldots,\,g_m\in G.
$$
Moreover, the map ${\bf Ch}_W^{W^{\otimes m}}\ni\Psi\mapsto f_\Psi\in C_S^0(G^N)$ is bijective.
\end{theorem}

\begin{corollary}
Channels $\Phi_1,\ldots,\,\Phi_m\in{\bf Ch}_W^W$ are compatible if and only if there is $f\in C_S^0(G^N)$ such that
$$
f_1(g)=f(g,e,\ldots,e),\ldots,\,f_m(g)=f(e,\ldots,e,g),\qquad g\in G,
$$
where $f_i:G\to\C$, $i=1,\ldots,\,m$, are the continuous functions of positive type such that $f_i(e)=1$ and $\Phi_i\big(W(g)\big)=f_i(g)W(g)$ for all $g\in G$, $i=1,\ldots,\,m$.
\end{corollary}

Let us give a few notes on the continuous non-compact phase space (${\bf X}=\R^N$) and retain the notation of Subsection \ref{sec:contphasespace}. Define the representation $W_m:\R^{2mN}\to\mc U(\hil^{\otimes m})$, $W_m(\vec{w}_1,\ldots,\vec{w}_m)=W(\vec{w}_1)\otimes\cdots\otimes W(\vec{w}_m)$, $\vec{w}_1,\ldots,\,\vec{w}_m\in\R^{2N}$.

\begin{theorem}
A covariant channel $\Psi\in{\bf Ch}_W^{W^{\otimes m}}$ is Gaussian if and only if there is a symmetric matrix $B\in\mc M_{2mN\times2mN}(\R)$, $B=(B_{k,l})_{k,l=1}^m$ where $B_{k,l}\in\mc M_{2N\times2N}(\R)$, $k,\,l=1,\ldots,\,m$, with
\begin{equation}
\left(\begin{array}{cccc}
B_{1,1}&B_{1,2}-i\Omega&\cdots&B_{1,m}-i\Omega\\
B_{1,2}^T+i\Omega^T&B_{2,2}&\cdots&B_{2,m}-i\Omega\\
\vdots&\vdots&\ddots&\vdots\\
B_{1,m}^T+i\Omega^T&B_{2,m}^T+i\Omega^T&\cdots&B_{m,m}
\end{array}\right)\geq0
\end{equation}
and a vector $\vec{c}\in\R^{2mN}$ (that can be picked at random) such that
$$
\Psi\big(W_m(\vec{z})\big)=e^{-\frac{1}{4}\vec{z}^TB\vec{z}-i\vec{c}^T\vec{z}}W(J_m^T\vec{z}),\qquad\vec{z}\in\R^{2mN},
$$
where $J_m:\R^{2N}\to\R^{2mN}$,
$$
\R^{2N}\ni\vec{w}\mapsto\underbrace{(\vec{w},\ldots,\vec{w})}_{m\ {\rm copies}}\in\R^{2mN}.
$$
\end{theorem}

\section{Conclusions}

We have studied covariant channels and their compatibility conditions and have shown that covariant channels always have a covariant joint channel. These general results were utilized in an analysis of the compatibility conditions of Weyl-covariant channels, i.e.,\ channels which behave symmetrically under phase space shifts. We have obtained necessary and sufficient conditions for compatibility of such channels involving the characteristic functions associated with the Weyl-covariant channels. Under some extra assumptions, we obtain a very descriptive compatibility condition which can be used as a recipe for generating all the covariant channels which are compatible with a fixed covariant channel. These results were investigated in the case of a non-compact continuous phase space and the case of a finite phase space with some illustrative examples. Finally, some notes on the multipartite case were presented.

We have seen that symmetries in the form of covariance properties can be used to greatly restrict the variety of joining of compatible channels, a results mirroring earlier results on quantum observables, particularly position and momentum \cite{Werner2004}. The analysis of Weyl covariance was chosen here for the simplicity arising from the canonical commutation relations. The case of more general physical symmetries, such as Euclidean covariance, remains to be examined. The non-commutative groups involved in these studies will provide richer structures and hopefully interesting methods of establishing incompatibility of quantum channels.

\section*{Acknowledgements}

The support in the form of discussions and suggestions of Teiko Heinosaari, Jukka Kiukas, and Juha-Pekka Pellonp\"a\"a has been instrumental for this work. Particularly Dr. Kiukas is recognized for helping to see the connection between Weyl-covariant joint channels and Gaussian channels. Henri Lyyra and Jose Teittinen are recognized for their help in plotting for Example \ref{ex:Zd} and for providing comic relieves during the author's visit at the University of Turku.


\begin{thebibliography}{99}

\bibitem{Busch86}
P.\ Busch, ``Unsharp reality and joint measurements for spin observables'', {\it Phys.\ Rev.\ D} {\bf 33}, pp. 2253-2261, (1986)

\bibitem{BuHeSchuSte2013}
P.\ Busch, T.\ Heinosaari, J.\ Schultz, and N.\ Stevens, ``Comparing the degrees of incompatibility inherent in probabilistic physical theories'', {\it Europhys. Lett.} {\bf 103}, 10002 (2013)

\bibitem{kirja}
P.\ Busch, P.\ Lahti,\ J.-P.\ Pellonp\"a\"a, and K.\ Ylinen, {\it ``Quantum Measurement''} (Springer 2016)

\bibitem{DaFuHo2006}
N.\ Datta, M.\ Fukuda, and A.\ S.\ Holevo, ``Complementarity and additivity for covariant channels'', {\it Quantum Information Proc.} {\bf 5}, pp. 170-207 (2006)

\bibitem{DeVaVe77}
B.\ Demoen, P.\ Vanheuverzwijn, and A.\ Verbeure, ``Completely positive maps on the CCR-algebra'', {\it Lett.\ Math.\ Phys.} {\bf 2}, pp. 161-166 (1977)

\bibitem{Dix}
J.\ Dixmier, {\it ``Von Neumann Algebras''} (North-Holland Publishing Company, Amsterdam - New York - Oxford, 1981)

\bibitem{GiCi2002}
G.\ Giedke and I.\ Cirac, ``Characterization of Gaussian operations and distillation of Gaussian states'', {\it Phys.\ Rev.\ A} {\bf 66}, 032316 (2002)

\bibitem{Haapasalo2015}
E.\ Haapasalo, ``Robustness of incompatibility for quantum devices, {\it J.\ Phys.\ A} {\bf 48}, 255303 (2015)

\bibitem{HaPe2017}
E.\ Haapasalo and J.-P.\ Pellonp\"a\"a, ``Covariant KSGNS construction and quantum instruments'', {\it Rev.\ Math.\ Phys.} {\bf 29}, 1750020 (2017)

\bibitem{HeKiSchu2015}
T.\ Heinosaari, J.\ Kiukas, and J.\ Schultz, ``Breaking Gaussian incompatibility on continuous variable quantum systems'', {\it J.\ Math.\ Phys.} {\bf 56}, 082202 (2015)

\bibitem{HeMi2017}
T.\ Heinosaari and T.\ Miyadera, ``Incompatibility of quantum channels'', {\it J.\ Phys.\ A: Math. Theor.} {\bf 50}, 135302 (2017)

\bibitem{HeReRyZi2018}
T.\ Heinosaari, D.\ Reitzner, T.\ Ryb\'ar, and M.\ Ziman, ``Incompatibility of unbiased qubit observables and Pauli channels'' {\it Phys.\ Rev.\ A} {\bf 97}, 022112 (2018)

\bibitem{HeReSta2008}
T.\ Heinosaari, D.\ Reitzner, and P.\ Stano, ``Notes on joint measurability of quantum observables'', {\it Found.\ Phys.} {\bf 38}, pp. 1133-1147 (2008)

\bibitem{Holevo2005}
A.\ S.\ Holevo, ``Additivity conjecture and covariant channels'', {\it Int.\ J.\ Quantum Information} {\bf 3}, pp. 41-47 (2005)

\bibitem{Lahti2003}
P.\ Lahti, ``Coexistence and joint measurability in quantum mechanics'' {\it Int.\ J.\ Theor.\ Phys.} {\bf 42}, pp. 893-906 (2003)

\bibitem{LaPu97}
P.\ Lahti and S.\ Pulmannov\'a, ``Coexistent observables and effects in quantum mechanics'', {\it Rep.\ Math.\ Phys.} {\bf 39}, pp. 339-351 (1997)

\bibitem{QuiVeBru2014}
M.T.\ Quintino, T.\ V\'ertesi, and N.\ Brunner, ``Joint measurability, Einstein-Podolsky-Rosen steering, and Bell nonlocality'', {\it Phys.\ Rev.\ Lett.} {\bf 113}, 160402 (2014)

\bibitem{GuMoUo2014}
R.\ Uola, T.\ Moroder, and O.\ G\"uhne, ``Joint measurability of GeneralizedMeasurements Implies Classicality'', {\it Phys.\ Rev.\ Lett.} {\bf 113}, 160403 (2014)

\bibitem{VaradarajanKirja}
V.S.\ Varadarajan, {\it ``Geometry of Quantum Theory''} (Springer, New York, 1985)

\bibitem{Werner84}
R.\ Werner, ``Quantum harmonic analysis on phase space'', {\it J.\ Math.\ Phys.} {\bf 25}, pp. 1404-1411 (1984)

\bibitem{Werner2004}
R.\ Werner, ``The uncertainty relation for joint measurement of position and momentum'', {\it Quant.\ Inf.\ Comp.} {\bf 4}, pp. 546-562 (2004)

\end{thebibliography}
\end{document}